\documentclass[11pt]{article}

\usepackage{fullpage}
\usepackage[utf8]{inputenc}
\usepackage[T1]{fontenc}
\usepackage{mathtools}
\usepackage{amsmath,amsfonts,amssymb,amsthm}
\usepackage{dsfont}
\usepackage{bbm}
\usepackage{bm}
\usepackage{mathrsfs}
\usepackage{nicefrac}
\usepackage{graphicx}
\usepackage{latexsym}
\usepackage{pgfplots}

\usepackage{array}
\usepackage{multirow}
\usepackage{booktabs}
\usepackage[shortlabels]{enumitem}
\usepackage{float}

\usepackage{thmtools}
\usepackage{thm-restate}
\usepackage{etoolbox}

\usepackage{verbatim}
\usepackage[strict]{changepage}
\usepackage[marginal]{footmisc}
\usepackage{url}
\usepackage{authblk}
\usepackage{optidef}

\allowdisplaybreaks

\definecolor{darkblue}{RGB}{0,76,156}
\definecolor{darkkblue}{RGB}{0,0,153}
\definecolor{blue2}{RGB}{102,178,255}
\definecolor{darkred}{RGB}{195,0,0}
\definecolor{refcolor}{rgb}{0.067,0.5,0.65}
\definecolor{urlcolor}{rgb}{0.1,0,0.9}

\usepackage[
    breaklinks,
    pdftex,
    colorlinks=true,
    linkcolor=refcolor,
    citecolor=refcolor,
    urlcolor=refcolor
]{hyperref}
\usepackage{cleveref}

\newtheorem{theorem}{Theorem}[section]
\newtheorem{definition}{Definition}
\newtheorem{proposition}[theorem]{Proposition}
\newtheorem{lemma}[theorem]{Lemma}
\newtheorem{corollary}[theorem]{Corollary}

\theoremstyle{definition}
\newtheorem{remark}[theorem]{Remark}

\mathchardef\ordinarycolon\mathcode`\:
\mathcode`\:=\string"8000
\def\vcentcolon{\mathrel{\mathop\ordinarycolon}}
\begingroup \catcode`\:=\active\lowercase{\endgroup\let :\vcentcolon}

\newcommand{\nc}{\newcommand}
\nc{\rnc}{\renewcommand}

\nc{\bra}[1]{\langle#1|}
\nc{\ket}[1]{|#1\rangle}
\nc{\ketbra}[2]{|#1\rangle\!\langle#2|}
\nc{\tr}{\operatorname{Tr}}
\nc{\ox}{\otimes}
\nc{\id}{\mathrm{id}}

\nc{\cA}{{\cal A}}
\nc{\cC}{{\cal C}}
\nc{\cD}{{\cal D}}
\nc{\cE}{{\cal E}}
\nc{\cF}{{\cal F}}
\nc{\cG}{{\cal G}}
\nc{\cH}{{\cal H}}
\nc{\cL}{{\cal L}}
\nc{\cM}{{\cal M}}
\nc{\cN}{{\cal N}}
\nc{\cS}{{\cal S}}
\nc{\cT}{{\cal T}}
\nc{\cU}{{\cal U}}
\nc{\cW}{{\cal W}}

\nc{\NN}{{\mathbb{N}}}

\nc{\idop}{{\mathds{1}}}
\nc{\SEP}{{\operatorname{SEP}}}
\nc{\PPT}{{\operatorname{PPT}}}
\nc{\LOCC}{{\operatorname{LOCC}}}
\nc{\Choi}{Choi-Jamio\l{}kowski }
\nc{\CPTP}{\text{\rm CPTP}}

\nc{\pT}{\mathsf{T}}
\nc{\bfA}{\mathbf{A}}
\nc{\bfB}{\mathbf{B}}
\nc{\suc}{\mathrm{suc}}

\begin{document}

\title{Entanglement cost of bipartite quantum channel discrimination under positive partial transpose operations}

\author[1]{Chengkai Zhu\thanks{C.\ Z.\ and S.\ H.\ contributed equally to this work. Email: czhu696@connect.hkust-gz.edu.cn}}
\author[1]{Shuyu He\thanks{she726@connect.hkust-gz.edu.cn}}
\author[2]{Gereon Ko\ss mann
\thanks{kossmann@physik.rwth-aachen.de}}
\author[1]{Xin Wang
\thanks{felixxinwang@hkust-gz.edu.cn}}
\affil[1]{\small Thrust of Artificial Intelligence, Information Hub,\par The Hong Kong University of Science and Technology (Guangzhou), Guangzhou 511453, China}
\affil[2]{RWTH Aachen University, Institute for Quantum Information, Germany}
\maketitle

\begin{abstract}
Quantum channel discrimination is a fundamental task in quantum information processing. In the one-shot regime, discrimination between two candidate channels is characterized by the diamond norm. Beyond this basic setting, however, many scenarios in distributed quantum information processing remain unresolved, motivating notions of distinguishability that capture the power of the available resources. In this work, we formulate a theory of testers for bipartite channel discrimination, leading to the concept of the entanglement cost of bipartite channel discrimination: the minimum Schmidt rank $k$ of a shared maximally entangled state required for local protocols to achieve the globally optimal success probability. We introduce $k$-injectable testers as a tester-based description of entanglement-assisted local discrimination and, in particular, study the class of $k$-injectable positive-partial-transpose (PPT) testers, which constitutes a numerically tractable relaxation of the practically relevant class of LOCC testers.

For every $k$, we derive a semidefinite program (SDP) for the optimal success probability, which in turn yields an efficiently computable one-shot PPT entanglement cost. To render these optimization problems numerically feasible, we prove a symmetry-reduction principle for covariant channel pairs, thereby reducing the effective dimension of the associated SDPs. Finally, by dualizing the SDP, we derive bounds on the composite channel-discrimination problem and illustrate our framework with proof-of-principle examples based on the depolarizing channel, the depolarized SWAP channel, and the Werner--Holevo channels.
\end{abstract}

\newpage
\tableofcontents

\section{Introduction}

Quantum channel discrimination is a fundamental primitive in quantum information processing, with a wide range of applications (see, e.g.,~\cite{Duan2007,Duan2008,Pirandola_2011,Wilde_2017g,Zhuang_2020} and subsequent works). The basic setting assumes an ensemble of quantum channels, each occurring with a specified probability, and allows the tester to use an arbitrary input state, possibly entangled with an untouched memory register, as well as an arbitrary measurement on the joint output register. The goal is to determine the optimal input state and measurement that maximize the probability of correctly identifying the implemented channel. In the special case of discriminating between two channels, this problem admits a well-known and widely applicable solution in terms of the \textit{diamond-norm distance}, first introduced by Kitaev~\cite{Kitaev_1997}. 

A closer look at the general setup of channel discrimination naturally raises the question of the role of the memory register, which may be viewed as providing entanglement assistance. The usefulness of entanglement in channel discrimination has been studied extensively (see, e.g.,~\cite{Childs_2000,Rosgen2004OnTH,Sacchi2005,Sacchi2005O,Berta_2013}). A seminal result of Piani and Watrous~\cite{Piani2009} shows that every entangled state is useful for some channel discrimination task. Beyond this qualitative advantage, it is also important to understand the \textit{quantitative} entanglement required, namely the minimum ancilla dimension sufficient for optimal discrimination. It is known that an ancilla with dimension equal to the \textit{input} dimension is always sufficient~\cite{Kitaev_1997}, and in some cases also necessary~\cite{Puzzuoli2017}. By contrast, an ancilla with dimension equal to the \textit{output} dimension is not always sufficient~\cite{Puzzuoli2017}. However, no general and efficient method is currently known for determining the exact ancilla dimension required for arbitrary pairs of channels.

\begin{figure}[t]
    \centering
    \includegraphics[width=1\linewidth]{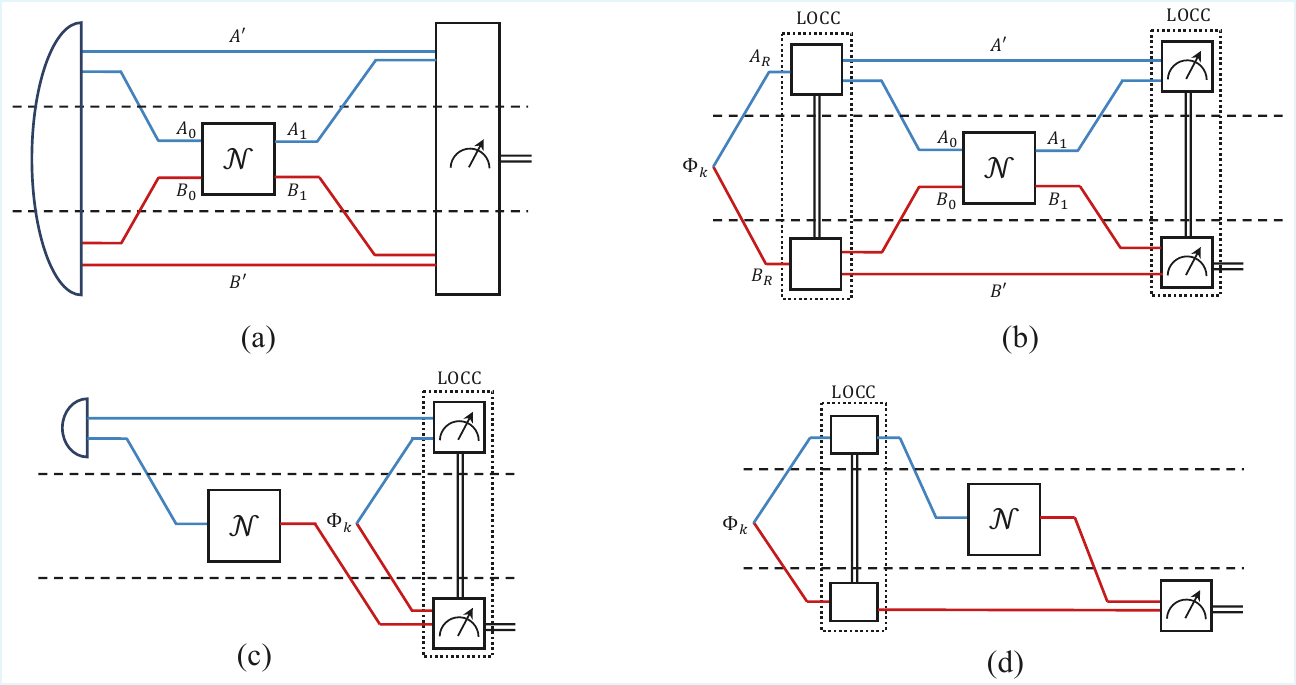}
    \caption{Protocols for entanglement-assisted discrimination of bipartite quantum channels.
    (a) A framework in which both parties prepared states and performed joint measurements to achieve diamond norm distance.
    (b) A general framework in which shared entanglement ($\Phi_k$) assists both the state-preparation and measurement stages. (c) A reduction to the point-to-point setting where entanglement specifically enhances the LOCC measurement. (d) A reduction in which entanglement is used to enhance probe-state preparation.}
    \label{fig:general_ent_assis}
\end{figure}

Moreover, many information theoretic tasks naturally involve \textit{bipartite quantum channels}, that is, channels with two input and two output registers that act jointly on inputs provided by both Alice and Bob, while restricting the allowed probe states and measurements to a practically relevant subclass such as local operations and classical communication (LOCC) or related classes~\cite{Chitambar_2014,CHILDS_2006}. In \Cref{fig:general_ent_assis}, we show the general case of a bipartite channel discrimination task. It is readily seen that the standard point-to-point channel is recovered as a special case when Bob's input and Alice's output are trivial one-dimensional systems. In the bipartite setting, if Alice and Bob, although spatially separated, are allowed to prepare entangled probe states and perform joint measurements, then the optimal performance in binary channel discrimination is again characterized by the diamond-norm distance, just as in the point-to-point case (see~\Cref{fig:general_ent_assis}~(a)). However, physically motivated constraints often limit Alice and Bob to LOCC in both probe-state preparation and measurement (see~\Cref{fig:locc_qcd}). Strategies subject to such locality constraints are known to be strictly less powerful than global ones~\cite{matthews2010}. Nevertheless, a quantitative understanding of the gap between local and global strategies for general bipartite channels remains largely lacking. This leads to the central question of this work:
\begin{center}
    \emph{How much shared entanglement is required to optimally discriminate a bipartite channel in the distributed setting via LOCC?}
\end{center}

We term this quantity the entanglement cost of bipartite channel discrimination. This quantity is particularly 
acute in a scenario of immense practical relevance to quantum networks and distributed quantum computing, where entanglement distribution is 
costly and must be optimized. Specifically, as illustrated in \Cref{fig:general_ent_assis}~(b), we aim to quantify the minimum dimension $k$ of a shared maximally entangled state $\Phi_k$, that allows Alice and Bob to achieve the optimal discrimination probability given by the diamond norm. In the point-to-point limit, our framework naturally recovers two settings: configuration (c) corresponds to utilizing entanglement to enhance the LOCC norm defined in Ref.~\cite{matthews2010} to the diamond norm, while configuration (d) corresponds to the problem of determining the minimum ancilla dimension required for the input probe to achieve the diamond norm.

\subsection{Contributions}
In this work, we introduce the \emph{entanglement cost of bipartite channel discrimination} and develop the mathematical framework of \emph{$k$-injectable testers} (cf.~Section~\ref{sec:bipartite_qcd_local}). These testers formalize discrimination strategies assisted by the shared entanglement of local dimension $k$. By additionally specifying the type of communication allowed in the final measurement stage, one obtains more refined classes, such as $k$-injectable PPT testers. This framework provides a systematic way to quantify the exact entanglement resources required for optimal channel discrimination under locality constraints (cf.~Section~\ref{sec:one_shot_entcost}). A natural problem is then to determine the entanglement cost and success probability for the class of $k$-injectable LOCC testers. However, since the set of LOCC operations is in general not even closed (see~\cite{Chitambar_2014}), it appears extremely difficult to characterize this class numerically. A tractable outer approximation is given by the set of separable testers, that is, testers with separable measurements, and a further outer approximation is provided by the set of PPT testers. Since characterizing separable testers essentially requires solving the separability problem, which is itself computationally hard (see, e.g.,~\cite{Fawzi2021_sep_states}), we focus in our numerical analysis on the set of PPT testers.

\begin{table}[b]
\centering
\caption{Summary of PPT entanglement costs for channel discrimination.}
\label{tab:ent_costs}
\renewcommand{\arraystretch}{1.3}
\begin{tabular}{@{} c c c @{}} 
\hline
\textbf{Channel pair} & \textbf{Dimension} & \textbf{Ent. cost} \\
\hline
Werner-Holevo & $d$ & $\log_2 d$ ebits \\
Point-to-point depolarizing & $d$ & 1 ebit\\
Bipartite depolarizing & $d_A \times d_B$ & 0 \\
Depolarized SWAP & $d$ & 1 ebit  \\
\hline
\end{tabular}
\end{table}

Focusing therefore on the class of PPT operations, which, as discussed above, provides a standard computable relaxation of LOCC, we prove that for fixed entanglement dimension $k$, the success probability can be computed via a semidefinite program (SDP) (cf.~Theorem~\ref{thm:sdp_k_inject_PPT}). This yields a tractable tool for bounding the performance of entanglement-assisted local strategies and, as a special case, recovers the previous analysis of the entanglement cost in bipartite quantum state discrimination~\cite{Zhu2025}. We further generalize our framework from discriminating individual channels to discriminating \emph{convex sets} of channels, thereby capturing scenarios in which the channels are only partially characterized. When both channel sets admit SDP representations, we derive an SDP for the entanglement-assisted worst-case optimal success probability in composite channel discrimination (cf.~Proposition~\ref{prop:composite_discrimination}).

As applications, we derive closed-form entanglement costs for highly symmetric channels (see~\Cref{tab:ent_costs}). Notably, we show that optimally distinguishing two bipartite depolarizing channels using PPT testers requires no entanglement. In striking contrast, the PPT entanglement cost is 1 ebit for any two distinct point-to-point depolarizing channels, regardless of the system's dimension. Furthermore, we establish that discriminating bipartite depolarized SWAP channels also incurs a PPT entanglement cost of 1 ebit across arbitrary dimensions. Finally, for a pair of $d$-dimensional Werner-Holevo channels, we prove that the PPT entanglement cost is $\log_2 d$ ebits.

\subsection{Organization}
The remainder of this paper is organized as follows. 
Section~\ref{sec:pre} establishes the notation and reviews the preliminaries of quantum channel discrimination and the diamond norm. In Section~\ref{sec:bipartite_qcd_local}, we formalize the bipartite discrimination task under locality constraints and introduce the hierarchy of testers (LOCC, SEP, and PPT). Section~\ref{sec:one_shot_entcost} contains our main theoretical results, defining the $k$-injectable tester and deriving the primal and dual SDPs for the entanglement-assisted success probability. In Section~\ref{sec:example}, we apply these tools to symmetric channels, providing exact calculations for the entanglement cost of channels of interest and numerical results.

\section{Preliminaries}\label{sec:pre}

\subsection{Notation}
We denote finite-dimensional Hilbert spaces associated with the systems of Alice and Bob as $\cH_A$ and $\cH_B$, respectively, with dimension $d_A$ and $d_B$. We abbreviate the tensor product space $A_{1}\ox A_{2}\ox\cdots\ox A_{n}$ as $A^{n}$ when all $A_{j}$ are isomorphic. The set of all linear operators on $\cH_{A}$ is denoted by $\cL(A)$, and the set of all quantum states, or density matrices, is denoted by $\cD(A)$. The identity operator on $\cH_A$ is denoted as $\idop_A$, and the maximally mixed state is denoted by $\pi_A = \idop_A/d_A$. The trace norm of a linear operator $X$ is defined as $\|X\|_1 = \tr\sqrt{X^\dagger X}$. Let $\{\ket{j}\}_{j = 0,\cdots,d-1}$ be the standard computational basis. A standard maximally entangled state of Schmidt rank $k$ is defined as $\Phi_k = 1/k \sum_{i,j = 0}^{k-1}\ketbra{ii}{jj}$. We will also use $\Phi_{AB}$ to explicitly denote the maximally entangled state on system $AB$. Measurements in quantum mechanics are described by Positive Operator-Valued Measures (POVM). A POVM on system $A$ is a set $\{ M_j \}_{j}$ of operators satisfying $M_j \geq 0$, and $\sum_j M_j = \idop_A$. In a bipartite quantum system $AB$, a state $\rho_{AB}$ is called separable (SEP) if it can be written as $\rho_{AB} = \sum_j p_j \rho_A^{(j)} \ox \sigma_B^{(j)}$ where $\{p_j\}$ is a probability distribution, and it is called entangled otherwise. A state is called a positive partial transpose (PPT) state if $\rho_{AB}^{\pT_B}\geq 0$ where $\pT_B$ denotes taking the transpose on the subsystem $B$. A POVM realized by local operations and classical communication (LOCC) is called a LOCC POVM. In order to clarify LOCC-operations in multipartite settings, we use the common notation $\operatorname{LOCC}(AA^\prime:B)$ to note LOCC-operations between $AA^\prime$ and $B$. For a POVM $\{ M_{AB}^{(j)}\}_{j}$, it is called an SEP POVM if $M_{AB}^{(j)}$ is separable for all $j$, and it is called a PPT POVM if $(M_{AB}^{(j)})^{\pT_B}\geq 0$ for all $j$.

A quantum channel is a completely positive trace-preserving (CPTP) map $\cN_{A\to B}: \cL(A) \rightarrow \cL(B)$. We denote by $\CPTP(A,B)$ the set of all quantum channels from $\cL(A)$ to $\cL(B)$. The \Choi operator of a quantum channel $\cN_{A\to B}$ is defined as $J_{A'B}^{\cN} \coloneqq \sum_{jk}\ketbra{j}{k}_{A'}\ox \cN_{A\to B}(\ketbra{j}{k})$ where $A'\cong A$. We denote by $\id_{A}$ the identity channel on $\cL(A)$. For a bipartite quantum channel $\cN_{A_0B_0\rightarrow A_1B_1}$, we assume that both input and output systems are shared between Alice and Bob, where we use the subscript `$0$' to denote the input system and the subscript `$1$' to denote the output system for each party.

For any two quantum processes, or, more generally, linear maps, they can be composed whenever the input system of one matches the output system of the other. After their connection, the Choi operator of the whole process can be derived through the \textit{link product}~\cite{Chiribella2008a}, denoted as $``\star"$. Formally, consider two linear maps $\cN_{A\to B} : \cL(A) \to \cL(B)$ and $\cM_{B\to C}: \cL(B) \to \cL(C)$. The Choi operator of the composite map $\cF_{A\to C} \coloneqq \cM_{B\to C} \circ \cN_{A\to B}$ from system $A$ to system $C$ is given by the link product of the individual Choi operators, defined by
\begin{equation}~\label{def:link_product}
J^{\cF}_{AC} = J^{\cN}_{AB} \star J^{\cM}_{BC} \coloneqq \tr_B \left[ \left( (J^{\cN}_{AB})^{\pT_B} \ox \idop_C \right) \left( \idop_A \ox J^{\cM}_{BC} \right) \right].
\end{equation}
The link product admits both associative and commutative properties, i.e.,
\begin{equation*}
\begin{aligned}
    J_{1} \star (J_{2} \star J_{3}) &= (J_{1} \star J_{2}) \star J_{3},\\
    J_{1} \star J_{2} &= J_{2} \star J_{1}.
\end{aligned}
\end{equation*}

\subsection{One-shot quantum channel discrimination}\label{sec:qcd}
In the task of minimum-error channel discrimination, one is given access to an unknown quantum channel $\cN_{A\to B}^{(j)}:\cL(A)\to\cL(B)$ that is known to have been drawn from a channel ensemble $\Omega = \{(p_j, \cN_{A\to B}^{(j)})\}_{j=0}^{m-1}$ with probability $p_j$. The goal is to determine which channel is given. To this end, a general strategy with only a single use of the unknown channel is to input a suitable probe state $\rho_{RA}$ with an arbitrary auxiliary system $R$ to the channel, and measure the output state with a POVM $\big\{M_{RB}^{(j)}\big\}_{j=0}^{m-1}$ on $RB$ system. By maximizing over all possible pairs of input state $\rho_{RA}$ and POVM $\big\{M_{RB}^{(j)}\big\}_{j=0}^{m-1}$ on $RB$, the maximal average success probability of discriminating $\big\{\cN_{A\to B}^{(j)}\big\}_{j=0}^{m-1}$ is then defined as (see,~e.g.,~\cite{Watrous2008} and subsequent work)
\begin{equation}\label{eq:general_success_probability_with_entanglement}
P_{\suc}(\Omega) \coloneqq \max_{\rho_{RA},\{M_j\}} \sum_{j=0}^{m-1} p_j \tr \Big[M_{RB}^{(j)} (\id_{R}\ox \cN_{A\to B}^{(j)})(\rho_{RA})\Big].
\end{equation}
When the set only consists of two quantum channels $\cN_0,\cN_1:\cL(A)\to \cL(B)$, the distinguishability of $\cN_0$ and $\cN_1$ and subsequently the average success probability in Eq.~(\ref{eq:general_success_probability_with_entanglement}) can be quantified by the distance induced by the \textit{diamond norm} (see,~e.g.,~\cite{Kitaev_1997}), i.e.,
\begin{equation*}
    \big\|\cN_{A\to B}^{(0)} - \cN_{A\to B}^{(1)}\big\|_{\Diamond} = \max_{\rho_{AR}\in\cD(RA)} \Big\|(\id_R\ox \cN_{A\to B}^{(0)})(\rho_{RA}) - (\id_R\ox \cN_{A\to B}^{(1)})(\rho_{RA})\Big\|_1,
\end{equation*}
where $R$ is an ancillary system with dimension $d_R \geq d_A$. The average success probability in Eq.~(\ref{eq:general_success_probability_with_entanglement}) can then be quantified as \cite{Watrous2008}
\begin{align}\label{eq:diamond_norm_as_success_probability}
    P_{\suc}\Big(\Big\{(\lambda,\cN_{A\to B}^{(0)}),(1-\lambda,\cN_{A\to B}^{(1)})\Big\}\Big) = \frac{1}{2}\big(1 + \big\| \lambda\, \cN_{A\to B}^{(0)}-(1-\lambda)\, \cN_{A\to B}^{(1)}\big\|_{\Diamond}\big).
\end{align}
If the two channels are given with equal probability, the maximal average success probability of discriminating $\cN_{A\to B}^{(0)}$ and $\cN_{A\to B}^{(1)}$ is according to Eq.~\eqref{eq:diamond_norm_as_success_probability} related to the diamond norm and has the following representation as a semidefinite program (SDP)~\cite{watrous2009}, involving the \Choi operators $J_{AB}^{\cN^{(j)}}$ of the underlying channels
\begin{equation}\label{Eq:Diamond}
\begin{aligned}
    P_{\suc}\Big(\Big\{(1/2,\cN_{A\to B}^{(0)}),(1/2,\cN_{A\to B}^{(1)})\Big\}\Big)  = \max &\; \frac{1}{2} + \frac{1}{2}\tr \Big[W_{AB}\big(J_{AB}^{\cN^{(0)}}-J_{AB}^{\cN^{(1)}}\big)\Big] \\
    {\rm s.t.}&\;\; 0\leq W_{AB} \leq \rho_{A}\ox\idop_{B}, \\
    &\;\; \rho_{A} \ge 0,~\tr\rho_{A} = 1.
\end{aligned}
\end{equation}
The corresponding dual problem is given by
\begin{equation}\label{sdp:sucprob_dual_diamond}
\begin{aligned}
    \min &\; \frac{1}{2} + \alpha\\
    {\rm s.t.} &\;\; C_{AB} \ge 0, ~\tr_{B}C_{AB} \le \alpha \idop_{A},\\
    &\;\; J_{AB}^{\cN} - J_{AB}^{\cM} \le  2C_{AB}.
\end{aligned}
\end{equation}
If we treat the probe state $\rho_A$ and the POVM $\big\{M_{AB}^{(j)}\big\}_{j=0}^{m-1}$ as a \textit{whole process}, a unified formalism of \textit{testers}~\cite{Chiribella2009}, also referred to as process POVMs~\cite{Ziman2008,Ziman2010}, is developed to tackle the problem of channel discrimination. Concretely, given a probe state $\rho_A \in \cD(A)$ and a POVM $\big\{M_{AB}^{(j)}\big\}_{j=0}^{m-1} \subseteq \cL(AB)$, the \Choi isomorphism associates to this choice a family of linear operators $\mathbf{T}=\big\{T_{AB}^{(j)}\big\}_{j=0}^{m-1}$, with $T_{AB}^{(j)}\in\cL(AB)$, called testers, satisfying $T_{AB}^{(j)} \geq 0$ for all $1\leq j \leq m$ such that 
\begin{equation}\label{eq:1tester_cond}
\sum_{j=0}^{m-1} T_{AB}^{(j)} = \rho_{A} \ox \idop_{B}.
\end{equation}
The maximal average success probability Eq.~\eqref{eq:general_success_probability_with_entanglement} for discriminating the ensemble $\big\{p_j,\cN_{A\to B}^{(j)}\big\}_{j=0}^{m-1}$ can then be rewritten as an optimization over testers as follows.
\begin{equation}\label{eq:1test_prob}
    P_{\suc}(\Omega) = \max_{\{T_{AB}^{(j)}\}} \sum_{j=0}^{m-1} p_j\tr\Big(T_{AB}^{(j)} J_{AB}^{\cN^{(j)}}\Big).
\end{equation}
Using this formalism, one can recover the SDP formulation in Eq.~\eqref{Eq:Diamond} from Eq.~\eqref{eq:1tester_cond} and~\eqref{eq:1test_prob} by observing that a binary POVM is uniquely characterized by a single positive effect $W_{AB}$.

\section{Bipartite channel discrimination under locality constraints}\label{sec:bipartite_qcd_local}
\begin{figure}[t]
    \centering
    \includegraphics[width=0.6\linewidth]{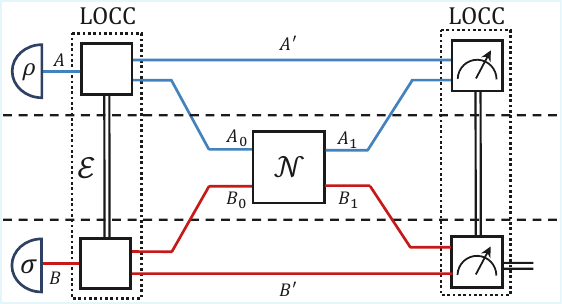}
    \caption{The general framework for bipartite channel discrimination is presented. Compared with the well-understood setting of channel discrimination in Section~\ref{sec:qcd}, several refinements are introduced. First, the channels in the ensemble $\Omega=\big\{(p_j,\cN_{A_0B_0\to A_1B_1}^{(j)})\big\}_j$ have bipartite inputs $A_0B_0$ and bipartite outputs $A_1B_1$, shared between Alice and Bob. Second, Alice and Bob are provided with local entanglement assistance in their laboratories, denoted by $A'$ and $B'$, respectively. By means of LOCC operations with respect to the bipartition $A\!:\!B$, they encode input states $\rho_A$ and $\sigma_B$ into states on $A_0A'$ and $B_0B'$. The unknown channel is then applied. Finally, restricting again to LOCC operations on $A_1A'\!:\!B_1B'$, Alice and Bob output a guess $j$, indicating that the channel $\cN_{A_0B_0\to A_1B_1}^{(j)}$ was applied.}
    \label{fig:locc_qcd}
\end{figure}

In this section, we introduce the task of bipartite quantum channel discrimination under locality constraints. This setting arises naturally in distributed quantum computing and quantum 
networks, where spatially separated parties cannot perform global operations. The central challenge is quantifying how much entanglement must be shared to overcome this locality restriction. To this end, we formalize the notion of $k$-injectable local testers.

\subsection{Problem formulation}
Specifically, consider the physically relevant scenario where two spatially separated parties, Alice and Bob, are restricted to LOCC~\cite{matthews2010}. This LOCC paradigm is also central to distributed quantum computing, offering a promising path toward scaling up quantum computations by linking separate quantum processors~\cite{peng2020simulating,mitarai2021constructing,piveteau2023circuit}. In this distributed scenario, the unknown channel is typically a bipartite $\cN_{A_0B_0\rightarrow A_1B_1}$ that serves as the communication link between Alice and Bob. For a binary discrimination, if the two parties are allowed to employ arbitrary shared entanglement as a resource, the optimal success probability for minimum-error discrimination is again characterized by the diamond norm as follows (see~\Cref{fig:general_ent_assis}~(a) for an operational illustration).
\begin{align*}
    P_{\suc}\Big(\Big\{(\lambda,\cN_{A_0B_0\rightarrow A_1B_1}^{(0)}),(1-\lambda,\cN_{A_0B_0\rightarrow A_1B_1}^{(1)})\Big\}\Big) = \frac{1}{2}\Big(1 + \Big\| \lambda\, \cN_{A_0B_0\rightarrow A_1B_1}^{(0)}-(1-\lambda)\, \cN_{A_0B_0\rightarrow A_1B_1}^{(1)}\Big\|_{\Diamond}\Big).
\end{align*}
In the LOCC setting considered here, however, the entanglement is constrained by the LOCC structure of the encoding and measurement operations. The corresponding LOCC discrimination protocol, as illustrated in~\Cref{fig:locc_qcd}, proceeds as follows.
\begin{itemize}
    \item[i).] Alice picks an initial state $\rho_{A} \in \cD(A)$, and Bob picks an initial state $\sigma_{B} \in \cD(B)$. They perform an encoding channel $\cE_{AB\to A'A_0B_0B'}$ that is $\operatorname{LOCC}(A\!:\!B)$ between Alice and Bob. Moreover, the encoding channel creates for Alice and Bob a memory system $A'$ and $B'$, respectively.
    \item[ii).] After applying the joint unknown channel $\cN^{(j)}_{A_0B_0\to A_1B_1}$ from the ensemble $\Omega$ to the output of the encoder $\cE_{AB\to A'A_0B_0B'}(\rho_{A}\ox \sigma_{B})$, Alice keeps the $A_1$-part and Bob keeps the $B_1$-part.
    \item[iii).] Alice and Bob perform an LOCC measurement $\operatorname{LOCC}(A_1A'\!:\!B_1B')$ on the state 
    \begin{align}
    \omega_{A_1A' B_1B'} \coloneqq \cN^{(j)}_{A_0 B_0\to A_1 B_1}\circ\cE_{AB\to A'A_0B_0B'}(\rho_{A}\ox \sigma_{B})    
    \end{align}
    whose $A_1A'$ system is held by Alice and $B_1B'$ system is held by Bob. Based on the measurement results for $\omega_{A_1A' B_1B'}$, Alice and Bob output a channel $\cN^{(j)}_{A_0B_0\to A_1B_1}$ from the ensemble $\Omega$.
\end{itemize}
This formulation allows us to compare the discrimination power of LOCC protocols with that of protocols utilizing global quantum operations or arbitrary shared entanglement.

\subsection{Tester under locality constraints}
Under the \Choi isomorphism, the discrimination process in~\Cref{fig:locc_qcd} can be described by the link product (see Eq.~\eqref{def:link_product}) of all input states and physical operations, excluding the unknown channel~\cite{Chiribella2009}. More precisely, for fixed input states $\rho_{A} \in \cD(A),\sigma_{B} \in \cD(B)$, encoding $\operatorname{LOCC}(A\!:\!B)$ channel $\cE_{AB\to A'A_0B_0 B'}$ and a $\operatorname{LOCC}(A_1A'\!:\!B_1B')$-POVM $\big\{M_{A'A_1B'B_1}^{(j)}\big\}_j$, we define the corresponding LOCC tester $\big\{T_{A_0B_0A_1B_1}^{(j)}\big\}_j$ of the channel discrimination task in~\Cref{fig:locc_qcd} as
\begin{equation}\label{eq:locc_tester}
    T_{A_0B_0A_1B_1}^{(j)} \coloneqq \rho_{A} \star \sigma_{B} \star J_{AB A_0B_0 A' B'}^{\cE} \star M_{A'A_1B'B_1}^{(j)}.
\end{equation}
Although an LOCC tester, constructed from an LOCC encoding $\cE$ and an LOCC POVM $\{M^{(j)}\}_j$, represents the most physically relevant class, characterizing such testers mathematically is notoriously difficult (see, e.g.,~\cite{Chitambar_2014}). 

To make the problem tractable, we consider relaxations of the LOCC constraint. A natural relaxation is to replace the LOCC encoding with SEP operations~\cite{Bennett1999,Rains1999}, i.e., channels with product Kraus operators between $A$ and $B$. Since $\cE_{AB\to A'A_0B_0 B'}$ is a SEP operation if and only if its \Choi operator is separable \cite{Horodecki_2003}, we can write
\begin{equation*}
    J_{AB A_0B_0 A' B'}^{\cE} = \sum_{x} p_x \alpha^{x}_{AA_0A'} \ox \beta^{x}_{BB_0B'}.
\end{equation*}
Meanwhile, we can relax the LOCC measurement $M^{(j)}$ to the SEP POVM, i.e.,
\begin{equation*}
    M_{A'A_1B'B_1}^{(j)} = \sum_y q_y \zeta^{y}_{A'A_1} \ox \mu^{y}_{B'B_1}.
\end{equation*}
Therefore, we arrive at a relaxation of the LOCC tester
\begin{equation*}
    \widehat{T}_{A_0B_0A_1B_1}^{(j)} = \sum_{x}\sum_{y} p_x q_y \tr_{A_0A_1}\Big[(\zeta_{A'A_1}^{y})^{\pT}\Big(\alpha^{x}_{AA_0A'}\rho_{A}^{\pT}\Big)\Big] \ox \tr_{B_0B_1}\Big[(\mu_{B'B_1}^{y})^{\pT}\Big(\beta^{x}_{BB_0B'}\sigma_{B}^{\pT}\Big)\Big].
\end{equation*}
Notice that $\widehat{T}_{A_0B_0A_1B_1}^{(j)}$ is a separable operator between $A_0A_1$ and $B_0B_1$, which motivates the definition of a SEP tester as follows.
\begin{definition}[SEP tester]
A set of linear operators $\big\{T^{(j)}_{A_0B_0 A_1B_1}\big\}_{j=0}^{m-1}$ is called an SEP tester if there exists a state $\rho_{A_0B_0}$ such that $T^{(j)}_{A_0B_0 A_1B_1} \geq 0,\sum_{j=0}^{m-1} T^{(j)}_{A_0B_0 A_1B_1} = \rho_{A_0B_0} \ox \idop_{A_1B_1}$, and each element admits a decomposition
\begin{equation*}
T^{(j)}_{A_0B_0 A_1B_1} = \sum_{k} E_{A_0A_1}^{(k)} \ox F_{B_0B_1}^{(k)},    
\end{equation*}
where $E_{A_0A_1}^{(k)}\geq 0$ and $F_{B_0B_1}^{(k)} \geq 0$ for all $k$.
\end{definition}

Unfortunately, even the set of separable testers is hard to describe, as it requires basically a solution to the separability problem, which is already known to be computationally intractable for quantum states \cite{Fawzi2021_sep_states}. A well-known relaxation, which is even SDP representable, i.e., can be incorporated in a semidefinite program, for the set of separable states, is the positive-partial-transpose (PPT) condition. It is based on the observation that every separable state needs to stay positive under local positive maps, which are not completely positive. The PPT is a well-known work-horse for outer approximations of the set of separable states, as it is positive but not completely positive. We thus take the definition of a SEP tester and relax its separability condition in the following definition for a PPT tester to a computationally tractable condition using the PPT map. 

\begin{definition}[PPT tester]
A set of linear operators $\big\{T^{(j)}_{A_0B_0 A_1B_1}\big\}_{j=0}^{m-1}$ is called a PPT tester if there exists a state $\rho_{A_0B_0}$ such that $T^{(j)}_{A_0B_0 A_1B_1} \geq 0, \sum_{j=0}^{m-1} T^{(j)}_{A_0B_0 A_1B_1} = \rho_{A_0B_0} \ox \idop_{A_1B_1}$, and each element satisfies
\begin{equation*}
    \big(T^{(j)}_{A_0B_0 A_1B_1}\big)^{\pT_{B_0B_1}} \geq 0.
\end{equation*}
\end{definition}
Concluding the discussion on different testers under locality constraints, we note the strict inclusions between the following sets
\begin{equation*}
    \text{LOCC testers} \subsetneq \text{SEP testers} \subsetneq \text{PPT testers} \subsetneq \text{All testers}.
\end{equation*}

\section{One-shot PPT entanglement cost for optimal discrimination}\label{sec:one_shot_entcost}

Given the setup in~\Cref{fig:locc_qcd}, it is natural to quantify the amount of additional shared entanglement required to match, under locality constraints, the average success probability of the general setting in Section~\ref{sec:qcd}. In this section, we present our main result, which quantitatively characterizes how much entanglement is necessary for a local strategy to attain the globally optimal average success probability for discriminating two bipartite quantum channels, namely, the success probability determined by the diamond-norm distance in Eq.~\eqref{Eq:Diamond}. Operationally, this yields for an ensemble of channels $\Omega=\big\{(p_j,\cN_{A_0B_0\to A_1B_1}^{(j)})\big\}_j$ a characterization in terms of the entanglement required to discriminate $\Omega$ optimally.

\subsection{Entanglement injectable testers}
To utilize entanglement, Alice and Bob can share a maximally entangled state $\Phi_k \in\cD(A_RB_R)$ to assist in the encoding and measurement process, as shown in~\Cref{fig:general_under_locc_kebit}. Henceforth, we denote by $A_R$ and $B_R$ the reference systems of Alice and Bob, respectively, which store their shared entanglement, and we assume $\cH_{A_R}\cong \cH_{B_R}$ with $d_{A_R}=d_{B_R}=k$. Then, a natural question is 
\begin{center}
\emph{What is the minimum dimension $k$ required for an LOCC (SEP/PPT) tester to achieve the optimal average success probability of quantum channel discrimination given in Eq.~\eqref{eq:1test_prob}?}
\end{center}

\begin{figure}[t]
    \centering
    \includegraphics[width=1\linewidth]{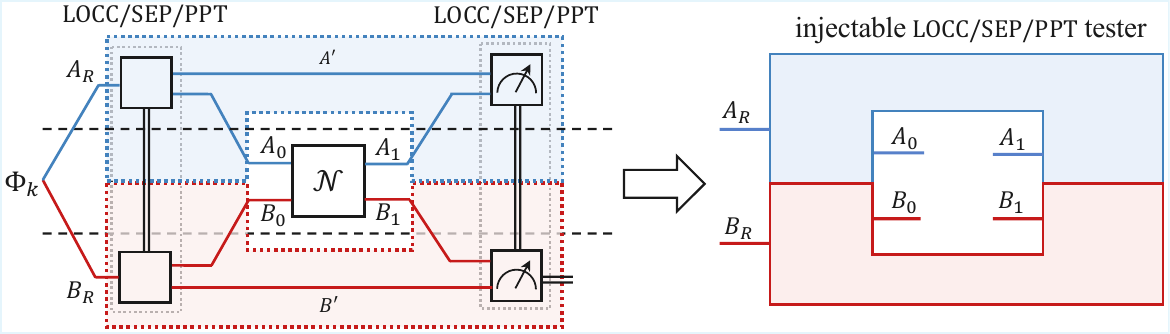}
    \caption{The adapted channel discrimination setup is shown. Comparing with Figure~\ref{fig:locc_qcd}, an additional resource of entanglement $\Phi_k\in\cD(A_RB_R)$ between Alice and Bob is added. In order to formally add them, we introduce the notion of $k$-injectable testers. It has explicit resource ports $A_RB_R$ such that, upon injecting $\Phi_k$, they induce an effective tester on the channel ports $A_0B_0A_1B_1$.
    }
    \label{fig:general_under_locc_kebit}
\end{figure}

To quantify the precise amount of entanglement required to enhance a tester under locality constraints to achieve global optimal performance, we introduce the concept of entanglement-injectable testers. Unlike conventional testers that act solely on the input channel, entanglement-injectable testers include two additional input systems, $A_R$ and $B_R$, which represent the interface for resource injection.

\begin{definition}[$k$-injectable PPT tester]~\label{def:k_injectable_PPT_def}
Let $\Big\{T_{k\text{-inj}}^{(j)}\Big\}_{j=0}^{m-1}$ be a set of linear operators where $T_{k\text{-inj}}^{(j)}\in\cL(A_RB_R A_0A_1 B_0B_1)$ and $d_{A_R} = d_{B_R} = k$. It is called a $k$-injectable PPT tester if
\begin{itemize}
    \item[i).] $T_{k\text{-inj}}^{(j)} \ge 0, ~\big(T_{k\text{-inj}}^{(j)}\big)^{\pT_{B_R B_0 B_1 }} \ge 0$ for all $j$,
    \item[ii).] $W_{A_RB_R A_0A_1 B_0B_1} = \tr_{A_1B_1} \left[W_{A_RB_RA_0A_1 B_0B_1} \right] \ox \frac{\idop_{A_1B_1}}{d_{A_1B_1}}$,
    \item[iii).] $W_{A_RB_R}  = d_{A_1B_1} \idop_{A_RB_R}$,
\end{itemize}
where $W_{A_RB_R A_0A_1 B_0B_1} = \sum_{j=0}^{m-1} T_{k\text{-inj}}^{(j)}$.
\end{definition}
We note that a $k$-injectable SEP tester can be defined in a similar way, where we require each $T_{k\text{-inj}}^{(j)}$ to be separable between $A_R A_0A_1$ and $B_R B_0B_1$. As we will mainly focus on the PPT case, which admits an SDP characterization, we only state the PPT case here. For a $k$-injectable PPT tester and any fixed injected state $\rho_{A_RB_R}$, the effective tester acting on the channel ports $A_0A_1$ and $B_0B_1$ is obtained by
\begin{equation}
    T_{A_0A_1 B_0B_1}^{(j)} \coloneqq \tr_{A_R B_R} \left[ T_{k\text{-inj}}^{(j)} (\rho_{A_RB_R}^{\pT} \ox \idop_{A_0A_1 B_0B_1}) \right].
\end{equation}

\subsection{Entanglement-assisted success probability and entanglement cost}
The entanglement-assisted average success probability of channel discrimination is then introduced as follows. Note that a similar idea has been employed to quantify the entanglement cost of quantum state discrimination in~\cite[Definition 4]{Zhu2025}.

\begin{definition}[Entanglement-assisted average success probability]\label{def:ent_ass_ave_suc_prob}
Given an ensemble of bipartite quantum channels $\Omega = \big\{(p_j,\cN^{(j)}_{A_0B_0\to A_1B_1})\big\}_{j=0}^{m-1}$ and an positive integer $k \ge 1$, the maximal $k$-ebit-assisted average success probability of discriminating $\Omega$ via a $k$-injectable $\cT$ tester is defined by
\begin{equation}\label{eq:max_k_ass_ave_suc_prob}
    P_{\suc,e}^{\cT}(\Omega; k) \coloneqq \max_{ \{T_{k\text{-inj}}^{(j)}\} \in \cT} \sum_{j=0}^{m-1} p_j\tr\left[T_{k\text{-inj}}^{(j)} \left(\Phi_k \ox J_{A_0A_1B_0B_1}^{\cN^{(j)}}\right)\right],
\end{equation}
where $\cT\in \{\LOCC, \SEP, \PPT\}$.
\end{definition}
This leads to the natural definition of the entanglement cost, representing the minimum resource dimension required to recover the performance of a global strategy.

\begin{definition}[One-shot entanglement cost of channel discrimination]
Given an ensemble of bipartite quantum channels $\Omega = \big\{(p_j,\cN^{(j)}_{A_0B_0\to A_1B_1})\big\}_{j=0}^{m-1}$, the one-shot $\cT$-entanglement cost of minimum-error discrimination is defined by
\begin{equation}
    E^{(1)}_{C,\cT}(\Omega) \coloneqq \min_{k\in \NN_+} \left\{\log k: P_{\suc,e}^{\cT}(\Omega; k) = P_{\suc}(\Omega)\right\}.
\end{equation}
\end{definition}

As our first result, we demonstrate that the maximal $k$-ebit-assisted average success probability of binary discrimination of $\{(\lambda,\cN_{A_0B_0\to A_1B_1}), (1-\lambda, \cM_{A_0B_0\to A_1B_1})\}$ via a $k$-injectable PPT tester can be written as an SDP shown in Theorem~\ref{thm:sdp_k_inject_PPT}. For simplicity, when there are only two channels, we use the notations $P_{\suc,e}^{\cT}(\cN,\cM; k, \lambda)$ to denote $P_{\suc,e}^{\cT}(\{(\lambda,\cN_{A_0B_0\to A_1B_1}), (1-\lambda, \cM_{A_0B_0\to A_1B_1})\})$ and $E^{(1)}_{C,\cT}(\cN,\cM;\lambda)$ to denote $E^{(1)}_{C,\cT}(\{(\lambda,\cN_{A_0B_0\to A_1B_1}), (1-\lambda, \cM_{A_0B_0\to A_1B_1})\})$. We will frequently use bold capital letter $\bfA$ to denote all subsystems that belong to Alice and use $\bfB$ to denote all subsystems that belong to Bob, e.g., $\bfA \coloneqq A_0A_1,\bfB \coloneqq B_0B_1$. 

\begin{theorem}\label{thm:sdp_k_inject_PPT}
Let $\cN_{A_0B_0\to A_1B_1}$ and $\cM_{A_0B_0\to A_1B_1}$ be two bipartite quantum channels given with prior probabilities $\lambda$ and $1-\lambda$, respectively, where $\lambda\in(0,1)$. The maximal $k$-ebit-assisted average success probability of discriminating $\cN_{A_0B_0\to A_1B_1}$ and $\cM_{A_0B_0\to A_1B_1}$ via a $k$-injectable PPT tester is given by
\begin{equation}\label{eq:sdp_k_ppt_inject}
\begin{aligned}
P^{\PPT}_{\suc,e}(\cN, \cM; k, \lambda) =\max &\; \lambda \tr (W_{\bfA\bfB}^{(0)}J_{\bfA\bfB}^{\cN}) + (1-\lambda)\tr (W_{\bfA\bfB}^{(1)}J_{\bfA\bfB}^{\cM})\\
{\rm s.t.}
&\;\; W_{\bfA\bfB}^{(j)} \geq 0,~Q_{\bfA\bfB}^{(j)}\geq 0, ~j \in \{0,1\},\\
&\;\; \tr (W_{\bfA\bfB}^{(0)}+W_{\bfA\bfB}^{(1)}) = \tr(Q_{\bfA\bfB}^{(0)}+Q_{\bfA\bfB}^{(1)}) = d_{A_1B_1},\\
&\;\; W^{(0)}_{\bfA\bfB}+W^{(1)}_{\bfA\bfB} = (W^{(0)}_{A_0B_0}+W^{(1)}_{A_0B_0})\ox \pi_{A_1B_1},\\
&\;\; Q^{(0)}_{\bfA\bfB}+Q^{(1)}_{\bfA\bfB} = (Q^{(0)}_{A_0B_0}+Q^{(1)}_{A_0B_0})\ox \pi_{A_1B_1},\\
&\;\; (1-k)(Q_{\bfA\bfB}^{(j)})^{\pT_{\bfB}} \le (W_{\bfA\bfB}^{(j)})^{\pT_{\bfB}} \le (1+k)(Q_{\bfA\bfB}^{(j)})^{\pT_{\bfB}}.
\end{aligned}
\end{equation}
\end{theorem}

\begin{proof}
Following Definition \ref{def:ent_ass_ave_suc_prob}, we specifically focus on the discrimination between the quantum channels $\cN_{A_0B_0\to A_1B_1}$ and $\cM_{A_0B_0\to A_1B_1}$ with prior probability $\lambda$ and $1-\lambda$. 
Precisely, we have that
\begin{subequations}
\begin{align}
P^{\PPT}_{\suc,e} &(\cN, \cM; k, \lambda)\notag \\
=  \max_{\Pi^{(0)},\Pi^{(1)}}& \ \lambda \tr \left[ \Pi^{(0)}_{A_RB_RA_0A_1B_0B_1} \big( \Phi_k\ox J_{A_0A_1B_0B_1}^{\cN} \big) \right] \notag\\
&\quad + (1-\lambda) \tr \left[ \Pi^{(1)}_{A_RB_R A_0A_1B_0B_1} \big( \Phi_k\ox J_{A_0A_1B_0B_1}^{\cM} \big) \right] \label{eq:proof_objective_func}\\
{\rm s.t.} &\;\; \Pi^{(0)}_{A_RB_R A_0A_1B_0B_1},~\Pi^{(1)}_{A_RB_R A_0A_1B_0B_1} \ge 0,\label{sdp:positivity}\\
&\;\; \Pi^{(0)}_{A_RB_R A_0A_1B_0B_1}+\Pi^{(1)}_{A_RB_R A_0A_1B_0B_1} = (\Pi^{(0)}_{A_RB_R A_0B_0} + \Pi^{(1)}_{A_RB_R A_0B_0})\ox \pi_{A_1B_1},\label{sdp:ns}\\
&\;\; \Pi^{(0)}_{A_RB_R}+\Pi^{(1)}_{A_RB_R} = d_{A_1B_1}\idop_{A_RB_R},\label{sdp:tp}\\
&\;\; (\Pi^{(0)}_{A_RB_R A_0A_1B_0B_1})^{\pT_{B_R B_0B_1}}, ~(\Pi^{(1)}_{A_RB_R A_0A_1B_0B_1})^{\pT_{B_R B_0B_1}} \ge 0,\label{sdp:ppt}
\end{align}
\end{subequations}
where $\{\Pi^{(j)}_{A_RB_R A_0A_1B_0B_1}\}_j$ is a $k$-injectable PPT testers. For any optimal solution $\Pi^{(j)}_{A_RB_R A_0A_1B_0B_1}$ ($j=0,1$), consider
\begin{equation*}
\widehat{\Pi}^{(j)}_{A_RB_R A_0A_1B_0B_1} \coloneqq (U_{A_R}\ox \overline{U}_{B_R}\ox \idop_{A_0A_1B_0B_1}) \Pi^{(j)}_{A_RB_R A_0A_1B_0B_1} (U_{A_R}\ox \overline{U}_{B_R}\ox \idop_{A_0A_1B_0B_1})^{\dagger},
\end{equation*}
where $\cH_{A_R}\cong \cH_{B_R}$ and $U_{A_R} = U_{B_R}$. We note that $\{\widehat{\Pi}^{(j)}_{A_RB_R A_0A_1B_0B_1}\}_j$ satisfies all conditions in Eq.~\eqref{sdp:positivity}~\eqref{sdp:ns}~\eqref{sdp:tp}~\eqref{sdp:ppt}, and yields the same objective value in Eq.~\eqref{eq:proof_objective_func}, since
\begin{equation*}
\begin{aligned}
    &\tr \left[ \widehat{\Pi}^{(j)}_{A_RB_R A_0A_1B_0B_1} \big( \Phi_k\ox J_{A_0A_1B_0B_1}^{\cN} \big) \right]\\
    =& \tr \left[ \Pi^{(j)}_{A_RB_R A_0A_1B_0B_1} \big( (U_{A_R}\ox \overline{U}_{B_R})\Phi_k(U_{A_R}\ox \overline{U}_{B_R})^{\dag} \ox J_{A_0A_1B_0B_1}^{\cN} \big) \right]\\
    =& \tr \left[ \Pi^{(j)}_{A_RB_R A_0A_1B_0B_1} \big( \Phi_k \ox J_{A_0A_1B_0B_1}^{\cN} \big) \right],
\end{aligned}
\end{equation*}
where we have leveraged the cyclicity of the trace in the first equality and the fact that $(U_{A_R}\ox \overline{U}_{B_R})\Phi_k(U_{A_R}\ox \overline{U}_{B_R})^{\dag} = \Phi_k$ for any unitary $U$ in the second equality. Hence, $\{\widehat{\Pi}^{(j)}_{A_RB_R A_0A_1B_0B_1}\}_j$ is also optimal. Furthermore, we note that any convex combination of optimal solutions remains optimal. With these two observations, we conclude that the following operator is also an optimal solution to the SDP
\begin{equation*}
\widetilde{\Pi}^{(j)}_{A_RB_R A_0A_1B_0B_1} \coloneqq \int dU (U_{A_R}\ox \overline{U}_{B_R}\ox \idop_{A_0A_1B_0B_1}) \Pi^{(j)}_{A_RB_R A_0A_1B_0B_1} (U_{A_R}\ox \overline{U}_{B_R}\ox \idop_{A_0A_1B_0B_1})^{\dagger},
\end{equation*}
where the integral is taken over the Haar measure of the unitary group. Notice that $\widetilde{\Pi}^{(j)}_{A_RB_R A_0A_1B_0B_1}$ commutes with all unitaries of the form $U\ox \overline{U}$ on $ \cH_{A_R}\ox \cH_{B_R}$ where $d_{A_R} = d_{B_R} = k$. and the properties of isotropic twirling~\cite{Horodecki1999,Watrous_2018}, the commutant of the group representation $U \mapsto U \ox \overline{U}$ on system $A_R B_R$ is spanned by the orthogonal projectors $\Phi_{A_R B_R}$ and $\idop_{A_R B_R} - \Phi_{A_R B_R}$. Furthermore, note that the unitary group acts trivially on $A_0A_1B_0B_1$, and the global commutant is strictly the tensor product of the local commutants~\cite[Chapter IV, Theorem 5.9]{Takesaki2002}. 
Consequently, we can consider the optimal solution $\widetilde{\Pi}^{(j)}_{A_RB_R A_0A_1B_0B_1}$ with a form
\begin{equation}
\widetilde{\Pi}^{(j)}_{A_RB_R A_0A_1B_0B_1} = \Phi_{A_RB_R} \ox W_{A_0A_1B_0B_1}^{(j)} + (\idop_{A_RB_R}-\Phi_{A_RB_R}) \ox Q_{A_0A_1B_0B_1}^{(j)}
\end{equation}
for some operators $W_{A_0A_1B_0B_1}^{(j)}$ and $Q_{A_0A_1B_0B_1}^{(j)}$. The objective function from Eq.~\eqref{eq:proof_objective_func} is rewritten as
\begin{equation*}
\lambda \tr (W_{A_0A_1B_0B_1}^{(0)}J_{A_0A_1B_0B_1}^{\cN}) + (1-\lambda)\tr (W_{A_0A_1B_0B_1}^{(1)}J_{A_0A_1B_0B_1}^{\cM}).
\end{equation*}
In the following, we will equivalently deduce the constraints on a $k$-injectable PPT tester, i.e., the constraints on $\Pi^{(j)}_{A_RB_R A_0A_1B_0B_1}$, to the constraints on $W_{A_0A_1B_0B_1}^{(j)}$ and $Q_{A_0A_1B_0B_1}^{(j)}$. 
For the positivity constraints in Eq.~\eqref{sdp:positivity}, we have 
\begin{equation}\label{eq:proof_cp_condition}
W_{A_0A_1B_0B_1}^{(j)} \geq 0,~Q_{A_0A_1B_0B_1}^{(j)}\geq 0.
\end{equation}
For the constraint in Eq.~\eqref{sdp:ns}, i.e.,
\begin{equation}
    \Pi^{(0)}_{A_RB_R A_0A_1B_0B_1}+\Pi^{(1)}_{A_RB_R A_0A_1B_0B_1} = (\Pi^{(0)}_{A_RB_R A_0B_0} + \Pi^{(1)}_{A_RB_R A_0B_0})\ox \pi_{A_1B_1},
\end{equation}
we have 
\begin{equation}\label{eq:proof_second_condition}
\begin{aligned}
&W^{(0)}_{A_0A_1B_0B_1}+W^{(1)}_{A_0A_1B_0B_1} = (W^{(0)}_{A_0B_0} + W^{(1)}_{A_0B_0})\ox \pi_{A_1B_1},\\
&Q^{(0)}_{A_0A_1B_0B_1}+Q^{(1)}_{A_0A_1B_0B_1} = (Q^{(0)}_{A_0B_0} + Q^{(1)}_{A_0B_0})\ox \pi_{A_1B_1}.
\end{aligned}
\end{equation}
For the constraint in Eq.~\eqref{sdp:tp}, i.e., $\Pi^{(0)}_{A_RB_R}+\Pi^{(1)}_{A_RB_R} = d_{A_1B_1}\idop_{A_RB_R}$, we have 
\begin{equation}\label{eq:proof_third_condition}
\tr (W_{A_0A_1B_0B_1}^{(0)}+W_{A_0A_1B_0B_1}^{(1)}) = \tr(Q_{A_0A_1B_0B_1}^{(0)}+Q_{A_0A_1B_0B_1}^{(1)}) = d_{A_1B_1}.
\end{equation}
For the PPT constraints in Eq.~\eqref{sdp:ppt}, i.e.,  $(\Pi^{(0)}_{A_RB_R A_0A_1B_0B_1})^{\pT_{B_R B_0B_1}}, ~(\Pi^{(1)}_{A_RB_R A_0A_1B_0B_1})^{\pT_{B_R B_0B_1}} \ge 0$, we denote $P_+$ and $P_-$ as the symmetric and anti-symmetric projections, respectively. From the spectral decomposition, we know that $\Phi_{A_RB_R}^{\pT_{B_R}} = (P_+ - P_-)/k$ and  
\begin{equation*}
\begin{aligned}
\Pi^{\pT_{B_RB_0B_1}} &= (W_{A_0A_1B_0B_1}^{(j)})^{\pT_{B_0B_1}} \ox \Phi_{A_RB_R}^{\pT_{B_R}} + (Q_{A_0A_1B_0B_1}^{(j)})^{\pT_{B_0B_1}}\ox (\idop_{A_RB_R} - \Phi_{A_RB_R})^{\pT_{B_R}}\\
&= (W_{A_0A_1B_0B_1}^{(j)})^{\pT_{B_0B_1}} \ox \frac{P_+ - P_-}{k} + (Q_{A_0A_1B_0B_1}^{(j)})^{\pT_{B_0B_1}} \ox \frac{(k-1)P_+ + (k+1)P_-}{k}\\
&= \Big[ (W_{A_0A_1B_0B_1}^{(j)})^{\pT_{B_0B_1}} + (k-1)(Q_{A_0A_1B_0B_1}^{(j)})^{\pT_{B_0B_1}}\Big] \ox \frac{P_+}{k}\\
&\quad + \Big[-(W_{A_0A_1B_0B_1}^{(j)})^{\pT_{B_0B_1}} +(k+1)(Q_{A_0A_1B_0B_1}^{(j)})^{\pT_{B_0B_1}}\Big] \ox \frac{P_-}{k}.
\end{aligned}
\end{equation*}
Since $P_+$ and $P_-$ are positive and orthogonal to each other, we have $\Pi_{A_RB_R A_0A_1B_0B_1}^{\pT_{B_R B_0B_1}} \geq 0$ if and only if 
\begin{equation}\label{eq:proof_forth_condition}
\begin{aligned}
(W_{A_0A_1B_0B_1}^{(j)})^{\pT_{B_0B_1}} + (k-1)(Q_{A_0A_1B_0B_1}^{(j)})^{\pT_{B_0B_1}} & \ge 0,\\
-(W_{A_0A_1B_0B_1}^{(j)})^{\pT_{B_0B_1}} + (k+1)(Q_{A_0A_1B_0B_1}^{(j)})^{\pT_{B_0B_1}} & \ge 0.
\end{aligned}
\end{equation}
Combining Eq.~\eqref{eq:proof_cp_condition},~\eqref{eq:proof_second_condition},~\eqref{eq:proof_third_condition}, and~\eqref{eq:proof_forth_condition}, the SDP can be simplified to
\begin{equation}
\begin{aligned}
\max &\; \lambda \tr (W_{\bfA\bfB}^{(0)}J_{\bfA\bfB}^{\cN}) + (1-\lambda)\tr (W_{\bfA\bfB}^{(1)}J_{\bfA\bfB}^{\cM})\\
{\rm s.t.}
&\;\; W_{\bfA\bfB}^{(j)} \geq 0,~Q_{\bfA\bfB}^{(j)}\geq 0, ~j \in \{0,1\},\\
&\;\; \tr (W_{\bfA\bfB}^{(0)}+W_{\bfA\bfB}^{(1)}) = \tr(Q_{\bfA\bfB}^{(0)}+Q_{\bfA\bfB}^{(1)}) = d_{A_1B_1},\\
&\;\; W^{(0)}_{\bfA\bfB}+W^{(1)}_{\bfA\bfB} = (W^{(0)}_{A_0B_0}+W^{(1)}_{A_0B_0})\ox \pi_{A_1B_1},\\
&\;\; Q^{(0)}_{\bfA\bfB}+Q^{(1)}_{\bfA\bfB} = (Q^{(0)}_{A_0B_0}+Q^{(1)}_{A_0B_0})\ox \pi_{A_1B_1},\\
&\;\; (1-k)(Q_{\bfA\bfB}^{(j)})^{\pT_{\bfB}} \le (W_{\bfA\bfB}^{(j)})^{\pT_{\bfB}} \le (1+k)(Q_{\bfA\bfB}^{(j)})^{\pT_{\bfB}}.
\end{aligned}
\end{equation}
Hence, we complete the proof.
\end{proof}

To simplify the SDP form of Eq.~\eqref{eq:sdp_k_ppt_inject}, we denote $\rho_{A_0B_0} \coloneqq (W_{A_0B_0}^{(0)} + W_{A_0B_0}^{(1)})/d_{A_1B_1}$ and $\sigma_{A_0B_0} \coloneqq (Q_{A_0B_0}^{(0)} + Q_{A_0B_0}^{(1)})/d_{A_1B_1}$. The variable $W_{A_0B_0}^{(1)}$ and $Q_{A_0B_0}^{(1)}$ can be eliminated by defining $W_{A_0A_1B_0B_1}^{(1)} \coloneqq \rho_{A_0B_0}\ox \idop_{A_1B_1} - W_{A_0A_1B_0B_1}^{(0)}$ and $Q_{A_0A_1B_0B_1}^{(1)} \coloneqq \sigma_{A_0B_0}\ox \idop_{A_1B_1} - Q_{A_0A_1B_0B_1}^{(0)}$. Thus, we have the following corollary.
\begin{corollary}
The maximal $k$-ebit-assisted success probability of discriminating $\cN_{A_0B_0\to A_1B_1}$ and $\cM_{A_0B_0\to A_1B_1}$ via a $k$-injectable PPT tester can be rewritten as 
\begin{equation}\label{sdp:kebit_sucprob}
\begin{aligned}
\max &\; 1-\lambda + \tr \Big[W_{\bfA\bfB}\Big(\lambda J_{\bfA\bfB}^{\cN} - (1-\lambda) J_{\bfA\bfB}^{\cM}\Big)\Big]\\
{\rm s.t.}
&\;\; \rho_{A_0B_0},\sigma_{A_0B_0}\geq 0,~\tr \rho_{A_0B_0} = \tr\sigma_{A_0B_0} = 1,\\
&\;\; 0\leq W_{\bfA\bfB}\leq \rho_{A_0B_0} \ox \idop_{A_1B_1},\\
&\;\; 0\leq Q_{\bfA\bfB}\leq \sigma_{A_0B_0} \ox \idop_{A_1B_1},\\
&\;\; (1-k)Q_{\bfA\bfB}^{\pT_{\bfB}} \le W_{\bfA\bfB}^{\pT_{\bfB}} \le (1+k)Q_{\bfA\bfB}^{\pT_{\bfB}}\\
&\;\; (1-k)(\sigma_{A_0B_0}^{\pT_{B_0}}\ox \idop_{A_1B_1} - Q_{\bfA\bfB}^{\pT_{\bfB}}) \le \rho_{A_0B_0}^{\pT_{B_0}}\ox\idop_{A_1B_1} - W_{\bfA\bfB}^{\pT_{\bfB}} \le (1+k)(\sigma_{A_0B_0}^{\pT_{B_0}}\ox\idop_{A_1B_1} - Q_{\bfA\bfB}^{\pT_{\bfB}}).
\end{aligned}
\end{equation}
\end{corollary}

\begin{remark}
We make two remarks on Theorem~\ref{thm:sdp_k_inject_PPT}. Firstly, when $\cN_{A_0B_0\to A_1B_1}$ and $\cM_{A_0B_0\to A_1B_1}$ are replacer channels, e.g., $\cN_{A_0B_0\to A_1B_1}(X) = \tr (X)\cdot \rho_{A_1B_1}$ and $\cM_{A_0B_0\to A_1B_1}(X) = \tr (X)\cdot \sigma_{A_1B_1}$ for any $X\in\cL(A_0B_0)$, our result here reduces to the case of bipartite quantum state discrimination. The SDP~\eqref{eq:sdp_k_ppt_inject} recovers the one in~\cite[Proposition 5]{Zhu2025}:
\begin{equation*}
\begin{aligned}
P^{\PPT}_{\suc,e}(\rho_{A_1B_1}, \sigma_{A_1B_1}; k, \lambda) = \max &\; \lambda \tr (W_{A_1B_1}^{(0)}\rho_{A_1B_1} + (1-\lambda)\tr (W_{A_1B_1}^{(1)}\sigma_{A_1B_1})\\
{\rm s.t.}
&\;\; W_{A_1B_1}^{(j)} \geq 0,~Q_{A_1B_1}^{(j)}\geq 0, ~j \in \{0,1\},\\
&\;\; W^{(0)}_{A_1B_1}+W^{(1)}_{A_1B_1} = \idop_{A_1B_1},\\
&\;\; Q^{(0)}_{A_1B_1}+Q^{(1)}_{A_1B_1} = \idop_{A_1B_1},\\
&\;\; (1-k)(Q_{A_1B_1}^{(j)})^{\pT_{B_1}} \le (W_{A_1B_1}^{(j)})^{\pT_{B_1}} \le (1+k)(Q_{A_1B_1}^{(j)})^{\pT_{B_1}}.
\end{aligned}
\end{equation*}
Secondly, when $B_0$ and $A_1$ are one-dimensional trivial subsystems, our result reduces to the case of point-to-point channel discrimination. The $k$-ebit-assisted average success probability is given as follows, where we abbreviate $A_0$ and $B_1$ by $A$ and $B$, respectively.
\begin{equation}\label{sdp:kebit_sucprob_pp}
\begin{aligned}
\max &\; 1-\lambda + \tr \Big[W_{AB}\Big(\lambda J_{AB}^{\cN} - (1-\lambda) J_{AB}^{\cM}\Big)\Big]\\
{\rm s.t.}
&\;\; \rho_{A},\sigma_{A}\geq 0,~\tr \rho_{A} = \tr\sigma_{A} = 1,\\
&\;\; 0\leq W_{AB}\leq \rho_{A} \ox \idop_{B},\\
&\;\; 0\leq Q_{AB}\leq \sigma_{A} \ox \idop_{B},\\
&\;\; (1-k)Q_{AB}^{\pT_B} \le W_{AB}^{\pT_B} \le (1+k)Q_{AB}^{\pT_B}\\
&\;\; (1-k)(\sigma_{A}\ox \idop_{B} - Q_{AB}^{\pT_{\bfB}}) \le \rho_{A}\ox\idop_{B} - W_{AB}^{\pT_{\bfB}} \le (1+k)(\sigma_A\ox\idop_{B} - Q_{AB}^{\pT_B}).
\end{aligned}
\end{equation}
When $B_0$ is a one-dimensional trivial subsystem, we obtain maximal success probability of discriminating broadcasting channels $\cN_{A_0\to A_1B_1}$ and $\cM_{A_0\to A_1B_1}$.
\begin{equation*}
\begin{aligned}
\max &\; 1-\lambda + \tr \Big[W_{\bfA B_1}\Big(\lambda J_{\bfA B_1}^{\cN} - (1-\lambda) J_{\bfA B_1}^{\cM}\Big)\Big]\\
{\rm s.t.}
&\;\; \rho_{A_0},\sigma_{A_0}\geq 0,~\tr \rho_{A_0} = \tr\sigma_{A_0} = 1,\\
&\;\; 0\leq W_{\bfA B_1}\leq \rho_{A_0} \ox \idop_{A_1B_1},\\
&\;\; 0\leq Q_{\bfA B_1}\leq \sigma_{A_0} \ox \idop_{A_1B_1},\\
&\;\; (1-k)Q_{\bfA B_1}^{\pT_{B_1}} \le W_{\bfA B_1}^{\pT_{\bfB}} \le (1+k)Q_{\bfA B_1}^{\pT_{B_1}}\\
&\;\; (1-k)(\sigma_{A_0}\ox \idop_{A_1B_1} - Q_{\bfA B_1}^{\pT_{B_1}}) \le \rho_{A_0}\ox\idop_{A_1B_1} - W_{\bfA B_1}^{\pT_{B_1}} \le (1+k)(\sigma_{A_0}\ox\idop_{A_1B_1} - Q_{\bfA B_1}^{\pT_{B_1}}).
\end{aligned}
\end{equation*}
\end{remark}

Building on the SDP formulation of the entanglement-assisted average success probability established in Theorem~\ref{thm:sdp_k_inject_PPT}, the one-shot entanglement cost is given by the minimal integer $k$ for which the optimal value of SDP~\eqref{sdp:kebit_sucprob} equals the diamond-norm distance, shown in the following theorem. Essentially, we seek the minimum PPT entanglement cost of a global tester that can achieve the optimal discrimination. Similar ideas have been used to study the entanglement cost of quantum channel simulation~\cite{Wang2023}.

\begin{theorem}\label{thm:one_shot_ent_cost}
Let $\cN_{A_0B_0\to A_1B_1}$ and $\cM_{A_0B_0\to A_1B_1}$ be two bipartite quantum channels given with prior probabilities $\lambda$ and $1-\lambda$, respectively, where $\lambda\in(0,1)$. The one-shot PPT entanglement cost of minimum-error discrimination is given by
\begin{equation*}
\begin{aligned}
E_{c,\PPT}^{(1)}\big(&\cN_{A_0B_0\to A_1B_1}, \cM_{A_0B_0\to A_1B_1};\lambda\big) \\
=\log\min &\; k\\
{\rm s.t.}
&\;\; \rho_{A_0B_0},\sigma_{A_0B_0}\geq 0,~\tr \rho_{A_0B_0} = \tr\sigma_{A_0B_0} = 1,\\
&\;\; 0\leq W_{\bfA\bfB}\leq \rho_{A_0B_0} \ox \idop_{A_1B_1},\\
&\;\; 0\leq Q_{\bfA\bfB}\leq \sigma_{A_0B_0} \ox \idop_{A_1B_1},\\
&\;\; (1-k)Q_{\bfA\bfB}^{\pT_{\bfB}} \le W_{\bfA\bfB}^{\pT_{\bfB}} \le (1+k)Q_{\bfA\bfB}^{\pT_{\bfB}},\\
&\;\; (1-k)(\sigma_{A_0B_0}^{\pT_{B_0}}\ox \idop_{A_1B_1} - Q_{\bfA\bfB}^{\pT_{\bfB}}) \le \rho_{A_0B_0}^{\pT_{B_0}}\ox\idop_{A_1B_1} - W_{\bfA\bfB}^{\pT_{\bfB}} \le (1+k)(\sigma_{A_0B_0}^{\pT_{B_0}}\ox\idop_{A_1B_1} - Q_{\bfA\bfB}^{\pT_{\bfB}}),\\
&\;\; \tr \Big[W_{\bfA\bfB}(\lambda J_{\bfA\bfB}^{\cN} - (1-\lambda) J_{\bfA\bfB}^{\cM})\Big] = \frac{1}{2}\big\|\lambda \cN_{A_0B_0\to A_1B_1}- (1-\lambda) \cM_{A_0B_0\to A_1B_1}\big\|_{\Diamond}.
\end{aligned}
\end{equation*}
\end{theorem}

Notably, the optimization problem in Theorem~\ref{thm:one_shot_ent_cost} is not a valid SDP since both $k$ and $Q_{\bfA\bfB}$ are variables and there are bilinear constraints like $W_{\bfA\bfB}^{\pT_{\bfB}} \le (1+k)Q_{\bfA\bfB}^{\pT_{\bfB}}$. However, the cost can be computed via an SDP hierarchy by discretely increasing $k$. For the convergence step of the hierarchy, we have the following upper bound on the one-shot PPT entanglement cost of minimum-error discrimination of two bipartite quantum channels.

\begin{proposition}\label{thm:upperbound}
Let $\cN_{A_0B_0\to A_1B_1}$ and $\cM_{A_0B_0\to A_1B_1}$ be two bipartite quantum channels given with prior probabilities $\lambda$ and $1-\lambda$, respectively, where $\lambda\in(0,1)$. The one-shot PPT entanglement cost of minimum-error discrimination satisfies
\begin{equation}
    E_{c,\PPT}^{(1)}\big(\cN_{A_0B_0\to A_1B_1},\cM_{A_0B_0\to A_1B_1};\lambda\big) \leq \log \big(2d_{A_0B_0}^2 d_{A_1B_1} - 1\big).
\end{equation}
\end{proposition}
\begin{proof}
Suppose $\big\{T^{(0)}_{\bfA\bfB},T^{(1)}_{\bfA\bfB}\big\}$ is an optimal tester for discriminating $\cN_{A_0B_0\to A_1B_1}$ and $\cM_{A_0B_0\to A_1B_1}$ as formulated in Eq.~\eqref{eq:1test_prob}, and $\mathrm{rank}\big(T^{(0)}_{\bfA\bfB}\big) = r_0, \mathrm{rank}\big(T^{(1)}_{\bfA\bfB}\big) = r_1$. Since $0\leq T^{(j)}_{\bfA\bfB} \leq \rho_{A_0B_0}\ox\idop_{A_1B_1}$, we have that $\big\|T^{(j)}_{\bfA\bfB}\big\|_{\infty}\leq \big\|\rho_{A_0B_0}\ox\idop_{A_1B_1}\big\|_{\infty} = 1$. Thus, we can write 
\begin{equation}
    T^{(j)}_{\bfA\bfB} = \sum_{l=1}^{r_j}\lambda_l^{(j)} \ketbra{\psi^{(j)}_l}{\psi^{(j)}_l}_{\bfA\bfB}, 
\end{equation}
where $\lambda_l^{(j)}\in(0,1]$ for all $l,j$, and have $\sum_{l=1}^{r_j}\lambda_l^{(j)}\leq d_{A_1B_1}$. We now bound the extremal eigenvalues of the partially transposed operator $(T^{(j)}_{\bfA\bfB})^{\pT_{\bfB}}$. Consider
\begin{equation}
\lambda_{\min}\Big((T^{(j)}_{\bfA\bfB})^{\pT_{\bfB}}\Big) \geq \sum_{l=1}^{r_j}\lambda_l^{(j)} \lambda_{\min}\Big(\ketbra{\psi^{(j)}_l}{\psi^{(j)}_l}_{\bfA\bfB}^{\pT_{\bfB}}\Big) \geq -\frac{1}{2}\sum_{l=1}^{r_j}\lambda_l^{(j)} \geq \max\left\{-\frac{d_{A_1B_1}}{2},-\frac{r_j}{2}\right\},
\end{equation}
where we have denoted $\lambda_{\min}(X)$ as the minimum eigenvalue of $X$ and have used the fact that $\lambda_{\min}(\sum_j A_j)\geq \sum_j\lambda_{\min}(A_j)$ for any Hermitian matrices $A_j$, and $\ketbra{\psi}{\psi}_{\bfA\bfB}^{\pT_{\bfB}}\in[-1/2,1]$~\cite{Rana2013}. Similarly, we have that
\begin{equation}
\lambda_{\max}\Big((T^{(j)}_{\bfA\bfB})^{\pT_{\bfB}}\Big) \leq \sum_{l=1}^{r_j}\lambda_l^{(j)} \lambda_{\max}\Big(\ketbra{\psi^{(j)}_l}{\psi^{(j)}_l}_{\bfA\bfB}^{\pT_{\bfB}}\Big) \leq \sum_{l=1}^{r_j}\lambda_l^{(j)} \leq \min\{d_{A_1B_1},r_j\}.
\end{equation}
Therefore, we have that
\begin{equation*}
    \max\left\{-d_{A_1B_1}d_{A_0B_0},-r_jd_{A_0B_0}\right\} \frac{\idop_{A_0B_0A_1B_1}}{2d_{A_0B_0}} \leq (T^{(j)}_{\bfA\bfB})^{\pT_{\bfB}} \leq \min\{2d_{A_1B_1}d_{A_0B_0},2r_jd_{A_0B_0}\} \frac{\idop_{A_0B_0A_1B_1}}{2d_{A_0B_0}}.
\end{equation*}
Since the rank $r_j$ is upper bounded by the full Hilbert space dimension $d_{A_0B_0}d_{A_1B_1}$, let $k_{\max} = 2d_{A_0B_0}^2 d_{A_1B_1} - 1$. By setting
\begin{equation*}
    W_{\bfA\bfB}^{(0)} = T^{(0)}_{\bfA\bfB}, \quad W_{\bfA\bfB}^{(1)} = T^{(1)}_{\bfA\bfB}, \quad Q_{\bfA\bfB}^{(0)} = Q_{\bfA\bfB}^{(1)} = \frac{\idop_{\bfA\bfB}}{2d_{A_0B_0}},
\end{equation*}
the bounding constraints $(1-k_{\max})Q_{\bfA\bfB}^{\pT_{\bfB}} \le W_{\bfA\bfB}^{\pT_{\bfB}} \le (1+k_{\max})Q_{\bfA\bfB}^{\pT_{\bfB}}$ are universally satisfied. Thus, this assignment provides a universally feasible solution to the SDP in Eq.~\eqref{eq:sdp_k_ppt_inject}, establishing the upper bound
\begin{equation*}
    E_{c,\PPT}^{(1)}\big(\cN_{A_0B_0\to A_1B_1},\cM_{A_0B_0\to A_1B_1};\lambda\big) \leq \log \big(2d_{A_0B_0}^2 d_{A_1B_1} - 1\big).
\end{equation*}
\end{proof}

\begin{corollary}
For two bipartite quantum channels $\cN_{A_0B_0\to A_1B_1}$ and $\cM_{A_0B_0\to A_1B_1}$ given with prior probability $\lambda$ and $1-\lambda$, respectively, where $\lambda \in (0,1)$, there exists an SDP hierarchy, parameterized by the entanglement dimension $k$, for the one-shot PPT entanglement cost of minimum-error discrimination. This hierarchy exactly converges at level $k\leq 2d_{A_0B_0}^2 d_{A_1B_1} - 1$.
\end{corollary}
\begin{proof}
Theorem~\ref{thm:one_shot_ent_cost} implies that the following SDP can measure the disparity
between the average success probability via global testers and the $k$-injectable PPT testers, for any fixed $k\in\NN_+$.
\begin{equation*}
\begin{aligned}
\min &\; t\\
{\rm s.t.}
&\;\; \rho_{A_0B_0},\sigma_{A_0B_0}\geq 0,~\tr \rho_{A_0B_0} = \tr\sigma_{A_0B_0} = 1,\\
&\;\; 0\leq W_{\bfA\bfB}\leq \rho_{A_0B_0} \ox \idop_{A_1B_1},\\
&\;\; 0\leq Q_{\bfA\bfB}\leq \sigma_{A_0B_0} \ox \idop_{A_1B_1},\\
&\;\; (1-k)Q_{\bfA\bfB}^{\pT_{\bfB}} \le W_{\bfA\bfB}^{\pT_{\bfB}} \le (1+k)Q_{\bfA\bfB}^{\pT_{\bfB}}\\
&\;\; (1-k)(\sigma_{A_0B_0}^{\pT_{B_0}}\ox \idop_{A_1B_1} - Q_{\bfA\bfB}^{\pT_{\bfB}}) \le \rho_{A_0B_0}^{\pT_{B_0}}\ox\idop_{A_1B_1} - W_{\bfA\bfB}^{\pT_{\bfB}} \le (1+k)(\sigma_{A_0B_0}^{\pT_{B_0}}\ox\idop_{A_1B_1} - Q_{\bfA\bfB}^{\pT_{\bfB}}),\\
&\;\; -t \leq \tr \Big[W_{\bfA\bfB}(\lambda J_{\bfA\bfB}^{\cN} - (1-\lambda) J_{\bfA\bfB}^{\cM})\Big] - \frac{1}{2}\big\|\lambda \cN_{A_0B_0\to A_1B_1}- (1-\lambda) \cM_{A_0B_0\to A_1B_1}\big\|_{\Diamond}\leq t.
\end{aligned}
\end{equation*}
We can observe that the minimum $k$ such that the optimal value of the above SDP is zero gives the value of the PPT entanglement cost. By Proposition~\ref{thm:upperbound}, the SDP hierarchy converges in at most $2d_{A_0B_0}^2 d_{A_1B_1} - 1$ steps.
\end{proof}

\subsection{Composite quantum channel discrimination}\label{sec:composite}
Up to this point, our analysis has assumed that the candidate quantum channels are perfectly known to the discriminating parties. In practical distributed quantum networks, however, channel dynamics are often subject to partial characterization, correlated noise, or systematic uncertainties. In such scenarios, the unknown channel is drawn from one of two possible sets of channels, e.g., $\Omega_0$ and $\Omega_1$, rather than being a single specific channel. The task of distinguishing between these \textit{sets} is known as \textit{composite quantum channel discrimination}~\cite{bergh2025composite}. 

In composite quantum channel discrimination, the discriminating parties (Alice and Bob) know the sets $\Omega_0$ and $\Omega_1$, but they do not know exactly which specific channels $\cN$ and $\cM$ from those sets will be handed to them. Because they must choose their input states, encodings, and POVMs, which together form the tester, before the unknown channel is applied, their choice of tester must be universal for the entire set. Therefore, given prior probabilities $\lambda$ and $1-\lambda$, when $k$ ebits of shared entanglement assist Alice and Bob, the optimal worst-case success probability of discriminating $\Omega_0$ and $\Omega_1$ is defined by
\begin{equation}\label{eq:composite_maxmin}
P_{\suc,e}^{\PPT}(\Omega_0,\Omega_1;k,\lambda) \coloneqq \max_{\{T^{(j)}_{k\text{-inj}}\} \in \PPT} \min_{\substack{\cN \in \Omega_0,\\ \cM \in \Omega_1}} \left[ \lambda \tr\big[T_{k\text{-inj}}^{(0)}(\Phi_k \ox J^{\cN}_{\bfA\bfB})\big] + (1-\lambda) \tr\big[T_{k\text{-inj}}^{(1)}(\Phi_k \ox J^{\cM}_{\bfA\bfB})\big] \right],
\end{equation}
where the minimization ranges over all possible channels in the two sets. Suppose both $\Omega_0$ and $\Omega_1$ are convex and compact. Notice that the set of $k$-injectable PPT testers is convex and compact, and the trace objective function above is bilinear, i.e., linear in $T_{k\text{-inj}}^{(j)}$ and linear in the channel. We can therefore use Sion's minimax theorem~\cite{Sion1958OnGM} to exchange the $\min$ and $\max$ first, and then leverage SDP duality to absorb the inner maximization into the outer minimization in Eq.~\eqref{eq:composite_maxmin}, resulting in a single, tractable SDP. To this end, we first derive the dual problem of~\eqref{sdp:kebit_sucprob_pp}.

Let us introduce positive semidefinite Lagrange multipliers $C_{\bfA\bfB}, D_{\bfA\bfB}, E_{\bfA\bfB}, G_{\bfA\bfB}, H_{\bfA\bfB}, K_{\bfA\bfB} \ge 0$ for the operator inequalities, and scalar multipliers $\alpha, \beta \in \mathbb{R}$ for the trace constraints on the marginals. The Lagrangian function of the primal problem is given by
\begin{equation}
\begin{aligned}
\cL &= 1 - \lambda + \tr\big[W_{\bfA\bfB}\big(\lambda J_{\bfA\bfB}^{\cN} - (1-\lambda)J_{\bfA\bfB}^{\cM}\big)\big]\\
&\quad + \langle C_{\bfA\bfB}, \rho_{A_0B_0} \ox \idop_{A_1B_1} - W_{\bfA\bfB} \rangle + \langle D_{\bfA\bfB}, \sigma_{A_0B_0} \ox \idop_{A_1B_1} - Q_{\bfA\bfB} \rangle\\
&\quad + \langle E_{\bfA\bfB}, W_{\bfA\bfB}^{T_{\bfB}} - (1-k)Q_{\bfA\bfB}^{\pT_{\bfB}} \rangle + \langle G_{\bfA\bfB}, (1+k)Q_{\bfA\bfB}^{\pT_{\bfB}} - W_{\bfA\bfB}^{T_{\bfB}} \rangle\\
&\quad + \langle H_{\bfA\bfB}, \rho_{A_0B_0}^{\pT_{B_0}} \ox \idop_{A_1B_1} - W_{\bfA\bfB}^{T_{\bfB}} - (1-k)(\sigma_{A_0B_0}^{\pT_{B_0}} \ox \idop_{A_1B_1} - Q_{AB}^{T_{\bfB}}) \rangle \\
&\quad + \langle K_{\bfA\bfB}, (1+k)(\sigma_{A_0B_0}^{\pT_{B_0}} \ox \idop_{A_1B_1} - Q_{\bfA\bfB}^{T_{\bfB}}) - (\rho_{A_0B_0}^{\pT_{B_0}} \ox \idop_{A_1B_1} - W_{\bfA\bfB}^{T_{\bfB}}) \rangle \\
&\quad + \alpha(1 - \tr\rho_{A_0B_0}) + \beta(1 - \tr\sigma_{A_0B_0})\\
&= 1 - \lambda + \alpha + \beta\\
&\quad + \Big\langle \lambda J_{\bfA\bfB}^{\cN} - (1-\lambda)J_{\bfA\bfB}^{\cM} - C_{\bfA\bfB} + E_{\bfA\bfB}^{T_{\bfB}} - G_{\bfA\bfB}^{T_{\bfB}} - H_{\bfA\bfB}^{T_{\bfB}} + K_{\bfA\bfB}^{T_{\bfB}}, W_{\bfA\bfB} \Big\rangle\\
&\quad + \Big\langle \!\!-D_{\bfA\bfB} - (1-k)E_{\bfA\bfB}^{T_{\bfB}} + (1+k)G_{\bfA\bfB}^{T_{\bfB}} + (1-k)H_{\bfA\bfB}^{T_{\bfB}} - (1+k)K_{\bfA\bfB}^{T_{\bfB}}, Q_{\bfA\bfB} \Big\rangle \\
&\quad + \Big\langle \tr_{A_1B_1}\big(C_{\bfA\bfB} + H_{\bfA\bfB}^{\pT_{B_0}} - K_{\bfA\bfB}^{\pT_{B_0}}\big) - \alpha \idop_{A_0B_0}, \rho_{A_0B_0} \Big\rangle \\
&\quad + \Big\langle \tr_{A_1B_1}\big( D_{\bfA\bfB} + (k-1)H^{\pT_{B_0}}_{\bfA\bfB} + (k+1)K^{\pT_{B_0}}_{\bfA\bfB}\big) - \beta 1_{A_0B_0}, \sigma_{A_0B_0} \Big\rangle.
\end{aligned}
\end{equation}
The dual function is obtained by taking the supremum over the primal variables $W_{\bfA\bfB}, Q_{\bfA\bfB}, \rho_{A_0B_0}, \sigma_{A_0B_0} \ge 0$. To ensure the supremum is finite, the operator coefficients attached to the primal variables must be negative semidefinite. This defines the feasible region of the dual problem. 
Setting all variables proportional to the identity operator provides a strictly feasible interior point for the primal problem~\eqref{sdp:kebit_sucprob_pp}. Consequently, strong duality holds by Slater's condition~\cite{Boyd_Vandenberghe_2004}, giving the exact dual SDP as follows.
\begin{equation}\label{sdp:kebit_sucprob_dual}
\begin{aligned}
    P^{\PPT}_{\suc,e}(\cN, \cM; k, \lambda) = \min &\;\; 1-\lambda + \alpha + \beta \\
    {\rm s.t.}
    &\;\; C_{\bfA\bfB},~D_{\bfA\bfB},~E_{\bfA\bfB},~G_{\bfA\bfB},~H_{\bfA\bfB},~K_{\bfA\bfB} \ge 0, \\
    &\;\; \lambda J_{\bfA\bfB}^{\cN} - (1-\lambda) J_{\bfA\bfB}^{\cM} - C_{\bfA\bfB} \le H_{\bfA\bfB}^{\pT_{\bfB}}- E_{\bfA\bfB}^{\pT_{\bfB}} + G_{\bfA\bfB}^{\pT_{\bfB}}-K_{\bfA\bfB}^{\pT_{\bfB}},\\
    &\;\; (1+k)(G_{\bfA\bfB}^{\pT_{\bfB}}-K_{\bfA\bfB}^{\pT_{\bfB}}) + (1-k)(H_{\bfA\bfB}^{\pT_{\bfB}}- E_{\bfA\bfB}^{\pT_{\bfB}}) \le D_{\bfA\bfB}, \\
    &\;\; \tr_{A_1B_1}\Big(C_{\bfA\bfB} + H_{\bfA\bfB}^{\pT_{B_0}} - K_{\bfA\bfB}^{\pT_{B_0}}\Big) \le \alpha \idop_{A_0B_0},\\
    &\;\; \tr_{A_1B_1} \Big[ D_{\bfA\bfB}+(k-1)H_{\bfA\bfB}^{\pT_{B_0}} +(k+1)K_{\bfA\bfB}^{\pT_{B_0}} \Big] \le \beta \idop_{A_0B_0}.
\end{aligned}
\end{equation}
By applying Sion's minimax theorem, the original max-min problem in Eq.~\eqref{eq:composite_maxmin} transforms into a joint minimization over both the worst-case channels and the dual variables, formalized as follows.
\begin{proposition}[SDP for composite bipartite channel discrimination] \label{prop:composite_discrimination}
Let $\Omega_0$ and $\Omega_1$ be two convex, compact sets of bipartite quantum channels with SDP representation, given with prior probabilities $\lambda$ and $1-\lambda$, respectively. The $k$-ebit-assisted optimal worst-case success probability of discriminating $\Omega_0$ and $\Omega_1$ via a $k$-injectable PPT tester is given by the following SDP.
\begin{equation}
\begin{aligned}
P_{\suc,e}^{\PPT}(\Omega_0,\Omega_1;k,\lambda) = \min &\;\; 1-\lambda + \alpha + \beta \\
    {\rm s.t.}
    &\;\; C_{\bfA\bfB},~D_{\bfA\bfB},~E_{\bfA\bfB},~G_{\bfA\bfB},~H_{\bfA\bfB},~K_{\bfA\bfB} \ge 0, \\
    &\;\; J_{\bfA\bfB}^{\cN} \in \Omega_0,~J_{\bfA\bfB}^{\cM} \in \Omega_1,\\
    &\;\; \lambda J_{\bfA\bfB}^{\cN} - (1-\lambda) J_{\bfA\bfB}^{\cM} \le H_{\bfA\bfB}^{\pT_{\bfB}}- E_{\bfA\bfB}^{\pT_{\bfB}} + G_{\bfA\bfB}^{\pT_{\bfB}}-K_{\bfA\bfB}^{\pT_{\bfB}} + C_{\bfA\bfB},\\
    &\;\; (1+k)(G_{\bfA\bfB}^{\pT_{\bfB}}-K_{\bfA\bfB}^{\pT_{\bfB}}) + (1-k)(H_{\bfA\bfB}^{\pT_{\bfB}}- E_{\bfA\bfB}^{\pT_{\bfB}}) \le D_{\bfA\bfB}, \\
    &\;\; \tr_{A_1B_1}\Big(C_{\bfA\bfB} + H_{\bfA\bfB}^{\pT_{B_0}} - K_{\bfA\bfB}^{\pT_{B_0}}\Big) \le \alpha \idop_{A_0B_0},\\
    &\;\; \tr_{A_1B_1} \Big[ D_{\bfA\bfB}+(k-1)H_{\bfA\bfB}^{\pT_{B_0}} +(k+1)K_{\bfA\bfB}^{\pT_{B_0}} \Big] \le \beta \idop_{A_0B_0}.
\end{aligned}
\end{equation}
\end{proposition}

\section{Symmetry and examples}\label{sec:example}

In this section, we transition from the general theoretical framework established in Section~\ref{sec:one_shot_entcost} to concrete applications.
Given channels $\cN_{A_0B_0\to A_1B_1}$ and $\cM_{A_0B_0\to A_1B_1}$ with associated \Choi operators $J_{\bfA \bfB}^{\cN}$ and $J_{\bfA \bfB}^{\cM}$ and prior distribution $(\lambda,1-\lambda)$, Theorem~\ref{thm:sdp_k_inject_PPT} states that the maximal success probability with $k$-injectable PPT tester can be computed by the following SDP.
\begin{equation}\label{eq:problem_statement_revisited}
\begin{aligned}
P^{\PPT}_{\suc,e}(\cN, \cM; k, \lambda) =\max &\; \lambda \tr (W_{\bfA\bfB}^{(0)}J_{\bfA\bfB}^{\cN}) + (1-\lambda)\tr (W_{\bfA\bfB}^{(1)}J_{\bfA\bfB}^{\cM})\\
{\rm s.t.}
&\;\; W_{\bfA\bfB}^{(j)} \geq 0,~Q_{\bfA\bfB}^{(j)}\geq 0, ~j \in \{0,1\},\\
&\;\; \tr (W_{\bfA\bfB}^{(0)}+W_{\bfA\bfB}^{(1)}) = \tr(Q_{\bfA\bfB}^{(0)}+Q_{\bfA\bfB}^{(1)}) = d_{A_1B_1},\\
&\;\; W^{(0)}_{\bfA\bfB}+W^{(1)}_{\bfA\bfB} = (W^{(0)}_{A_0B_0}+W^{(1)}_{A_0B_0})\ox \pi_{A_1B_1},\\
&\;\; Q^{(0)}_{\bfA\bfB}+Q^{(1)}_{\bfA\bfB} = (Q^{(0)}_{A_0B_0}+Q^{(1)}_{A_0B_0})\ox \pi_{A_1B_1},\\
&\;\; (1-k)(Q_{\bfA\bfB}^{(j)})^{\pT_{\bfB}} \le (W_{\bfA\bfB}^{(j)})^{\pT_{\bfB}} \le (1+k)(Q_{\bfA\bfB}^{(j)})^{\pT_{\bfB}}.
\end{aligned}
\end{equation}
Generally, the problem size grows with the input and output dimension of the channels and thus becomes very large very quickly. However, in applications, we have many additional structures of the channels $\cN_{A_0B_0\to A_1B_1}$ and $\cM_{A_0B_0\to A_1B_1}$ at our disposal. 
In Section~\ref{sec:cov}, we demonstrate how covariance properties of quantum channels allow us to collapse the SDP into small, highly tractable linear programs. We then apply this tool to three fundamental channel models: depolarizing and depolarized SWAP channels. Then, in Section~\ref{sec:werner}, we study another class of fundamental channels, the Werner-Holevo channels.

\subsection{Covariant channel}\label{sec:cov}

For a finite group $\cG$ and every $g\in\cG$, let $g \mapsto U_{A_0}(g)$ and $g \mapsto V_{B_0}(g)$ be unitary representations acting on the input spaces of Alice and Bob, respectively. Assuming the input and output spaces are isomorphic, i.e., $\cH_{A_0}\cong \cH_{A_1}$ and $\cH_{B_0}\cong \cH_{B_1}$, a bipartite quantum channel $\cN_{A_0B_0 \to A_1B_1}$ is said to be 
\begin{itemize}
\item $(U_{A_0}, V_{B_0})$-covariant if for all $g\in \cG$
\begin{equation}\label{eq:def_covariant_channels}
\begin{aligned}
&\cN_{A_0B_0\to A_1B_1}\Big(\big(U_{A_0}(g)\ox V_{B_0}(g)\big) \rho_{A_0B_0} \big(U_{A_0}^\dag (g)\ox V_{B_0}^\dag(g)\big)\Big) \\
= &\;\Big(U_{A_1}(g)\ox V_{B_1}(g)\Big) \cN_{A_0B_0\to A_1B_1}(\rho_{A_0B_0}) \Big(U_{A_1}^\dag (g)\ox V_{B_1}^\dag(g)\Big),\quad \forall \rho_{A_0B_0}\in\cD(A_0B_0).
\end{aligned}
\end{equation}
\item $(U_{A_0}, V_{B_0})$-cross-covariant if for all $g\in \cG$
\begin{equation}\label{eq:def_cro_covariant_channels}
\begin{aligned}
&\cN_{A_0B_0\to A_1B_1}\Big(\big(U_{A_0}(g)\ox V_{B_0}(g)\big) \rho_{A_0B_0} \big(U_{A_0}^\dag (g)\ox V_{B_0}^\dag(g)\big)\Big) \\
= &\;\Big(V_{A_1}(g)\ox U_{B_1}(g)\Big) \cN_{A_0B_0\to A_1B_1}(\rho_{A_0B_0}) \Big(V_{A_1}^\dag (g)\ox U_{B_1}^\dag(g)\Big),\quad \forall \rho_{A_0B_0}\in\cD(A_0B_0).
\end{aligned}
\end{equation}
\end{itemize}
A crucial consequence of this operational symmetry is observed in the \Choi representation of the channel, i.e., the covariance of the map $\cN$ imposes a specific unitary invariance on its Choi operator.

\begin{lemma}\label{lem:covariant_choi}
Let $\cG$ be a finite group with unitary representations $U_{A_0}$ and $V_{B_0}$. 
$\cN_{A_0B_0 \to A_1B_1}$ is $(U_{A_0}, V_{B_0})$-covariant if and only if its Choi operator satisfies
\begin{equation*}
\Big(\overline{U}_{A_0}(g)\ox \overline{V}_{B_0}(g)\ox U_{A_1}(g)\ox V_{B_1}(g) \Big) J_{A_0B_0A_1B_1}^{\cN} \Big(\overline{U}^\dag_{A_0}(g)\ox \overline{V}^\dag_{B_0}(g)\ox U^\dag_{A_1}(g)\ox V^\dag_{B_1}(g) \Big) = J_{A_0B_0A_1B_1}^{\cN},
\end{equation*}
for all $g\in\cG$ and $\cN_{A_0B_0 \to A_1B_1}$ is $(U_{A_0}, V_{B_0})$-cross-covariant if and only if
\begin{equation*}
\Big(\overline{U}_{A_0}(g)\ox \overline{V}_{B_0}(g)\ox V_{A_1}(g)\ox U_{B_1}(g) \Big) J_{A_0B_0A_1B_1}^{\cN} \Big(\overline{U}^\dag_{A_0}(g)\ox \overline{V}^\dag_{B_0}(g)\ox V^\dag_{A_1}(g)\ox U^\dag_{B_1}(g) \Big) = J_{A_0B_0A_1B_1}^{\cN},
\end{equation*}
for all $g\in\cG$, where $\overline{V}$ denotes the complex conjugate of $V$.
\end{lemma}
\begin{proof}
    This readily follows from the definition of the \Choi operator and Eq.~\eqref{eq:def_covariant_channels}.
\end{proof}

When $B_0$ and $A_1$ are one-dimensional trivial subsystems, the channel covariance reduces to~\cite[Eq.~(2.11)]{Wilde_2017} 
\[
    \cN_{A_0\to B_1}(U_{A_0}(g)\rho_{A_0}U_{A_0}^\dag(g)) = V_{B_1}(g) \cN_{A_0\to B_1} V_{B_1}^\dag(g).
\]
We now apply this framework to the task of discriminating between two covariant channels, $\cN_{\bfA \to \bfB}$ and $\cM_{\bfA \to \bfB}$. Specifically, we consider the scenario where both channels are covariant with respect to the same compact group $\cG \subseteq \cU(\cH)$. Leveraging the invariance of the \Choi operators established in Lemma~\ref{lem:covariant_choi}, we can exploit the inherent symmetry of the optimization problem. This allows us to simplify the general SDP formulation in Eq.~\eqref{sdp:kebit_sucprob} by restricting the feasible region to the subspace of invariant operators without loss of generality.

\begin{proposition}[Symmetry reduction under channel covariance]\label{prop:sym_redu_covariance}
Let $\cG$ be a compact group with Haar measure $\mu$ and two unitary representations $g \mapsto U_{A_0}(g)$ on $\cH_{A_0}$ and $g\mapsto V_{B_0}(g)$ on $\cH_{B_0}$. Let two channels $\cN_{A_0B_0 \to A_1B_1}$ and $\cM_{A_0B_0 \to A_1B_1}$ be $(U_{A_0}, V_{B_0})$-covariant, i.e., for all $g\in\cG$ and $\rho \in \cD(A_0B_0)$
\begin{equation*}
\cC\Big(\big(U_{A_0}(g)\ox V_{B_0}(g)\big) \rho \big(U_{A_0}^\dag (g)\ox V_{B_0}^\dag(g)\big)\Big)= \Big(U_{A_1}(g)\ox V_{B_1}(g)\Big) \cC(\rho) \Big(U_{A_1}^\dag (g)\ox V_{B_1}^\dag(g)\Big), \forall\cC\in\{\cN,\cM\},
\end{equation*}
and $\cH_{A_0}\cong \cH_{A_1}, \cH_{B_0}\cong\cH_{B_1}$.
Define 
\begin{equation}\label{eq:gammaAB}
\Gamma_{\bfA\bfB}(g) \coloneqq \overline{U}_{A_0}(g)\ox \overline{V}_{B_0}(g) \ox U_{A_1}(g) \ox V_{B_1}(g).    
\end{equation}
Then the optimal value of SDP~\eqref{sdp:kebit_sucprob} is unchanged if one adds the invariance constraints: for all $g\in\cG$:
\begin{itemize}
\item[i).] $\Gamma_{\bfA\bfB}(g) W_{\bfA\bfB} \Gamma_{\bfA\bfB}(g)^\dagger = W_{\bfA\bfB}$,
\item[ii).] $\Gamma_{\bfA\bfB}(g) Q_{\bfA\bfB} \Gamma_{\bfA\bfB}(g)^\dagger = Q_{\bfA\bfB}$,
\item[iii).] $\Big(\overline{U}_{A_0}(g)\ox \overline{V}_{B_0}(g)\Big) \rho_{A_0B_0} \Big(\overline{U}_{A_0}(g)\ox \overline{V}_{B_0}(g)\Big)^\dagger = \rho_{A_0B_0}$,
\item[iv).] $\Big(\overline{U}_{A_0}(g)\ox \overline{V}_{B_0}(g)\Big) \sigma_{A_0B_0} \Big(\overline{U}_{A_0}(g)\ox \overline{V}_{B_0}(g)\Big)^\dagger = \sigma_{A_0B_0}$.
\end{itemize}
\end{proposition}
\begin{proof}
Define the twirls as follows
\begin{equation*}
\cS_{\bfA\bfB}(X) \!\coloneqq\!\! \int_G \Gamma_{\bfA\bfB}(g) X \Gamma_{\bfA\bfB}^\dagger(g) \, d\mu(g), ~
\cS_{A_0 B_0}(Y) \!\coloneqq\!\! \int_G \Big(\overline{U}_{A_0}(g)\ox \overline{V}_{B_0}(g)\Big) Y \Big(\overline{U}_{A_0}(g)\ox \overline{V}_{B_0}(g)\Big)^\dagger \, d\mu(g).
\end{equation*}
Both maps are \emph{unital}, \emph{completely positive}, \emph{trace preserving}, self-adjoint w.r.t.\ the Hilbert--Schmidt inner product, and idempotent. By Lemma~\ref{lem:covariant_choi}, we have that $\cS_{\bfA\bfB}(J^{\cN}_{\bfA\bfB} - J^{\cM}_{\bfA\bfB}) = J^{\cN}_{\bfA\bfB} - J^{\cM}_{\bfA\bfB}$. Therefore, for any feasible $(W_{\bfA\bfB},Q_{\bfA\bfB},\rho_{A_0B_0},\sigma_{A_0B_0})$ of SDP~\eqref{sdp:kebit_sucprob}, we have
\begin{equation*}
\tr\Big[ W_{\bfA\bfB} (J^{\cN}_{\bfA\bfB} - J^{\cM}_{\bfA\bfB}) \Big]
= \tr\Big[ \cS_{\bfA\bfB}(W_{\bfA\bfB}) \, \cS_{\bfA\bfB}(J^{\cN}_{\bfA\bfB} - J^{\cM}_{\bfA\bfB}) \Big]
= \tr\Big[ \cS_{\bfA\bfB}(W_{\bfA\bfB}) \, (J^{\cN}_{\bfA\bfB} - J^{\cM}_{\bfA\bfB}) \Big],
\end{equation*}
i.e., $\cS_{\bfA\bfB}(W_{\bfA\bfB})$ yields the same objective value. Now, let
\begin{equation}
\bar{W}_{\bfA\bfB} \coloneqq \cS_{\bfA\bfB}(W_{\bfA\bfB}),~
\bar{Q}_{\bfA\bfB} \coloneqq \cS_{\bfA\bfB}(Q_{\bfA\bfB}),~
\bar{\rho}_{A_0B_0} \coloneqq \cS_{A_0B_0}(\rho_{A_0B_0}),~
\bar{\sigma}_{A_0B_0} \coloneqq \cS_{A_0B_0}(\sigma_{A_0B_0}).
\end{equation}
Since the positivity and operator order are preserved under unitary conjugation and convex averaging, it follows that
\begin{equation}
0 \leq \bar{W}_{\bfA\bfB} \leq \bar{\rho}_{A_0B_0} \ox \idop_{A_1B_1},\quad
0 \leq \bar{Q}_{\bfA\bfB} \leq \bar{\sigma}_{A_0B_0} \ox \idop_{A_1B_1},
\end{equation}
and $\tr \bar{\rho}_{A_0B_0} = \tr \bar{\sigma}_{A_0B_0} = 1$. For the $\pT_{\bfB}$-constraints, notice that
\begin{equation}
\begin{aligned}
&\Big[\Big( \overline{U}_{A_0}(g)\ox \overline{V}_{B_0}(g) \ox U_{A_1}(g) \ox V_{B_1}(g)\Big) X \Big(\overline{U}^\dag_{A_0}(g)\ox \overline{V}^\dag_{B_0}(g) \ox U^\dag_{A_1}(g) \ox V^\dag_{B_1}(g) \Big)\Big]^{\pT_{B_0B_1}}\\
=&\; \Big(\overline{U}^\dag_{A_0}(g)\ox V_{B_0}(g) \ox U_{A_1}^\dag(g) \ox \overline{V}_{B_1}(g) \Big) X^{\pT_{B_0B_1}} \Big( \overline{U}_{A_0}(g)\ox V^\dag_{B_0}(g) \ox U_{A_1}(g) \ox V_{B_1}^\dag (g)\Big).
\end{aligned}
\end{equation}
Conjugating $(1 - k) Q_{\bfA\bfB}^{\pT_{\bfB}} \leq W_{\bfA\bfB}^{\pT_{\bfB}} \leq (1 + k) Q_{\bfA\bfB}^{\pT_{\bfB}}$ by $\Big(\overline{U}^\dag_{A_0}(g)\ox V_{B_0}(g) \ox U_{A_1}^\dag(g) \ox \overline{V}_{B_1}(g) \Big)$ and averaging over $g$ yields
\begin{equation}
(1 - k) \bar{Q}_{\bfA\bfB}^{\pT_{\bfB}} \leq \bar{W}_{\bfA\bfB}^{\pT_{\bfB}} \leq (1 + k) \bar{Q}_{\bfA\bfB}^{\pT_{\bfB}},
\end{equation}
because we can exchange the order of the integral and $\pT_{B_0B_1}$. The linearity of $\cS_{\bfA\bfB}$ gives, similarly,
\begin{equation}
(1 - k)(\bar{\sigma}^{\pT_{B_0}}_{A_0B_0} \ox \idop_{A_1B_1} - \bar{Q}_{\bfA\bfB}^{\pT_{\bfB}}) \leq \bar{\rho}_{A_0B_0}^{\pT_{B_0}} \ox \idop_{A_1B_1} - \bar{W}_{\bfA\bfB}^{\pT_{\bfB}} \leq (1 + k)(\bar{\sigma}_{A_0B_0}^{\pT_{B_0}} \ox \idop_{A_1B_1} - \bar{Q}_{\bfA\bfB}^{\pT_{\bfB}}).
\end{equation}
Thus, $(\bar{W}, \bar{Q}, \bar{\rho}, \bar{\sigma})$ is also feasible of SDP~\eqref{sdp:kebit_sucprob}. By construction, $\bar{W}, \bar{Q}$ satisfy $\Gamma_{\bfA\bfB}(g) (\cdot) \Gamma_{\bfA\bfB}^\dagger(g) = (\cdot)$ and $\bar{\rho}, \bar{\sigma}$ satisfy $\Big(\overline{U}_{A_0}(g)\ox \overline{V}_{B_0}(g)\Big) (\cdot) \Big(\overline{U}_{A_0}(g)\ox \overline{V}_{B_0}(g)\Big)^\dagger = (\cdot)$ for all $g$. Since we know that the objective value is unchanged by replacing $(W_{AB},Q_{AB},\rho_{A},\sigma_{A})$ with $(\bar{W}, \bar{Q}, \bar{\rho}, \bar{\sigma})$, the supremum is the same when restricting to $G$-invariant variables, and an optimal solution may be chosen invariant.
\end{proof}

We note that a parallel symmetry reduction can be derived for a pair of $(U_{A_0},V_{B_0})$-cross-covariant channels by choosing $\Gamma_{\bfA\bfB} \coloneqq \overline{U}_{A_0}(g)\ox \overline{V}_{B_0}(g) \ox V_{A_1}(g) \ox U_{B_1}(g)$ in Eq.~\eqref{eq:gammaAB}. A direct application of Proposition~\ref{prop:sym_redu_covariance} is that for any pair of bipartite classical channels $\cN$ and $\cM$, the one-shot PPT entanglement cost of minimum-error discrimination is $0$ ebits.

\begin{proposition}[Classical channel]\label{prop:classical_torus}
Fix orthonormal product bases
$\{\ket{a_0}\}$ on $A_0$, $\{\ket{b_0}\}$ on $B_0$, $\{\ket{a_1}\}$ on $A_1$, and $\{\ket{b_1}\}$ on $B_1$, assume the partial transpose $\pT_{\bfB}$ is taken with respect to the chosen basis on ${\bfB}=B_0B_1$.
Let $\cN_{A_0B_0\to A_1B_1}$ and $\cM_{A_0B_0\to A_1B_1}$ be classical with respect to these bases, i.e., their Choi operators are diagonal in the induced basis on $A_0B_0A_1B_1$. Then for every $k\in\NN_+$ and $\lambda\in(0,1)$, SDP~\eqref{sdp:kebit_sucprob} admits an optimal solution
$(W_{\bfA\bfB},Q_{\bfA\bfB},\rho_{A_0B_0},\sigma_{A_0B_0})$
for which all four variables are diagonal in the above bases.
Consequently, for such an optimal solution, one has
\begin{equation}\label{eq:diag_pt_identity}
X^{\pT_{\bfB}} = X,~\forall X\in\{W_{\bfA\bfB},Q_{\bfA\bfB}\},~~
Y^{\pT_{B_0}} = Y,~\forall Y\in\{\rho_{A_0B_0},\sigma_{A_0B_0}\}.
\end{equation}
In particular, at $k=1$, we have $P^{\PPT}_{\suc,e}(\cN,\cM;1,\lambda)=P_{\suc}(\cN,\cM;\lambda)$ and hence,
\begin{equation*}
E^{(1)}_{c,\PPT}(\cN,\cM;\lambda)=0.
\end{equation*}
\end{proposition}
\begin{proof}
    See Appendix~\ref{appendix:proof_class}.
\end{proof}
The utility of Proposition~\ref{prop:sym_redu_covariance} as a general tool for any channel pairs with covariance properties is demonstrated through its application to two concrete examples: the depolarizing channel and the depolarized SWAP channel.

\subsubsection{Depolarizing channel}\label{sec:depo}
As a canonical example of covariant channels, we consider the depolarizing channel, which serves as the standard theoretical model for uniform quantum noise. We first investigate the bipartite setting, where the noise acts globally across the joint input of Alice and Bob, and subsequently, the point-to-point setting.

For a bipartite input system $A_0B_0$ and a bipartite output system $A_1B_1$ where $\cH_{A_0}\cong \cH_{A_1}, \cH_{B_0}\cong \cH_{B_1}$, a bipartite depolarizing channel $\cD^{p}_{A_0B_0\to A_1B_1}(\cdot)$ with noise parameter $p$ is defined by 
\begin{equation}
    \cD^{p}_{A_0B_0\to A_1B_1}(X) = (1-p)X + p\tr (X)\frac{\idop_{A_1B_1}}{d_{A_1B_1}}.
\end{equation}
The Choi operators of the depolarizing channels are given by
\begin{equation}
    J_{\bfA\bfB}^{\cD^p} = d_{A_0}d_{B_0}(1-p)\Phi_{d_{A_0}}\ox \Phi_{d_{B_0}} + p\frac{\idop_{\bfA\bfB}}{d_{\bfA\bfB}}.
\end{equation}
By applying the general symmetry reduction established in Proposition~\ref{prop:sym_redu_covariance}, we can significantly simplify the optimization complexity. The full SDP formulation reduces to a tractable linear program over the isotropic subspace, as detailed in the following proposition.

\begin{proposition}[Symmetry reduction for bipartite depolarizing channels]\label{prop:bi_depo_LP}
Let $\cD^{p}_{A_0B_0\to A_1B_1}$ and $\cD_{A_0B_0\to A_1B_1}^{q}$ be two bipartite depolarizing channels with distinct noise parameters $p$ and $q$, respectively. The maximal $k$-ebit-assisted success probability of discriminating $\cD^{p}_{A_0B_0\to A_1B_1},\cD_{A_0B_0\to A_1B_1}^{q}$ via a $k$-injectable PPT tester is given by the following linear program.
\begin{equation}
\begin{aligned}
P^{\PPT}_{\suc,e}(\cD^{p},\cD^{q}; k, \frac{1}{2}) =\max &\;\frac{1}{2} + \frac{1}{2}(p - q) \left(d_{A_0B_0} - \frac{1}{d_{A_0B_0}}\right) (x - y)\\
{\rm s.t.}
&\;\; x, y, u, v \in [0, 1/d_{A_0B_0}],\\
&\;\; (1 - k) u \leq x \leq (1 + k) u,\\
&\;\; (1 - k) v \leq y \leq (1 + k) v,\\
&\;\; (1 - k)\left( \frac{1}{d_{A_0B_0}} - u \right) \leq \frac{1}{d_{A_0B_0}} - x \leq (1 + k)\left( \frac{1}{d_{A_0B_0}} - u \right),\\
&\;\; (1 - k)\left( \frac{1}{d_{A_0B_0}} - v \right) \leq \frac{1}{d_{A_0B_0}} - y \leq (1 + k)\left( \frac{1}{d_{A_0B_0}} - v \right).
\end{aligned}
\end{equation}
\end{proposition}
\begin{proof}
See Appendix~\ref{appendix:proof_bi_depo}.
\end{proof}

As a direct application of this symmetry reduction, we can evaluate the composite channel discrimination task (c.f.~Section~\ref{sec:composite}) for sets of depolarizing channels. Suppose the noise parameters of the channels are not perfectly characterized, but are known to belong to two disjoint compact sets $P$ and $Q$. The composite discrimination problem is simply reduced to minimizing the scalar distance between the two parameter sets.

\begin{corollary}[Composite discrimination of depolarizing channels]Let $\Omega_0 = \{\cD^p_{A_0B_0\to A_1B_1} : p \in P\}$ and $\Omega_1 = \{\cD^q_{A_0B_0\to A_1B_1} : q \in Q\}$ be two sets of bipartite depolarizing channels defined by compact parameter sets $P, Q \subset [0, 1]$ such that $p > q$ for all $p \in P$ and $q \in Q$. The $k$-ebit-assisted optimal worst-case success probability of discriminating $\Omega_0$ and $\Omega_1$ via a $k$-injectable PPT tester is given by 
\begin{equation}
P_{\suc, e}^{\PPT}\Big(\Omega_0, \Omega_1; k, \frac{1}{2}\Big) = \max_{x,y} \left[ \frac{1}{2} + \frac{1}{2}(\Delta_{P,Q}) \left(d_{A_0B_0} - \frac{1}{d_{A_0B_0}}\right) (x - y) \right]
\end{equation}
subject to the same constraints on $(x, y, u, v)$ as in Proposition~\ref{prop:bi_depo_LP}, where $\Delta_{P,Q} \coloneqq \min_{p \in P, q \in Q} (p - q)$ is the minimum noise gap between the two channel sets.
\end{corollary}
\begin{proof}
By Proposition~\ref{prop:bi_depo_LP}, we have
\begin{equation}
P_{\suc, e}^{\PPT}\Big(\Omega_0, \Omega_1; k, \frac{1}{2}\Big) = \max_{(x,y) \in \cD_k} \min_{p \in P, q \in Q} \left[ \frac{1}{2} + \frac{1}{2}(p - q) \left(d_{A_0B_0} - \frac{1}{d_{A_0B_0}}\right) (x - y) \right],
\end{equation}
where $\mathcal{D}_k$ is the feasible region of the scalar variables $(x, y, u, v)$ parameterized by $k$ and $d_{A_0B_0}$. Because $\cD_k$ is independent of $P$ and $Q$, and the expression $\left(d_{A_0B_0} - \frac{1}{d_{A_0B_0}}\right) (x - y)$ is strictly non-negative for any optimal choice of $x$ and $y$, the inner minimization strictly evaluates to minimizing $(p-q)$ over $P \times Q$. Taking $\Delta_{P,Q} = \min_{p \in P, q \in Q} (p - q)$ completes the proof.
\end{proof}

Using the linear program reduction in Proposition~\ref{prop:bi_depo_LP}, we show that for a pair of bipartite depolarizing channels, the one-shot PPT entanglement cost of minimum-error discrimination is $0$ ebits.

\begin{proposition}[PPT entanglement cost of bipartite depolarizing channels]\label{prop:en_cost_one_shot_bip_depo}
For any bipartite depolarizing channels $\cD^{p}_{A_0B_0\to A_1B_1}, \cD_{A_0B_0\to A_1B_1}^{q}$ with distinct noise parameters $p,q$ and equiprobable priors, the one-shot PPT entanglement cost of minimum-error discrimination is $0$ ebits.
\end{proposition}

\begin{proof}
From Proposition~\ref{prop:bi_depo_LP}, when $k = 1$, let 
\begin{equation}
    x \coloneqq \frac{1}{d_{A_0}d_{B_0}}, \quad y \coloneqq 0, \quad u = v \coloneqq \frac{1}{2d_{A_0}d_{B_0}}.
\end{equation}
By simple calculation, these solutions are feasible. Thus we have 
\begin{equation*}
    P^{\PPT}_{\suc,e}\Big(\cD^{p}_{A_0B_0\to A_1B_1},\cD_{A_0B_0\to A_1B_1}^{q}; k, \frac{1}{2}\Big) \ge \frac{1}{2} +  \frac{(p - q)(d_{A_0}^2d_{B_0}^2-1)}{2d_{A_0}^2d_{B_0}^2}.
\end{equation*}
Consider the dual SDP~\eqref{sdp:sucprob_dual_diamond}, let
\begin{equation}
    \alpha \coloneqq \frac{(p - q)(d_{A_0}^2d_{B_0}^2-1)}{2d_{A_0}^2d_{B_0}^2}, \quad C_{\bfA\bfB} \coloneqq \frac{(p - q)(\idop_{\bfA\bfB}-\Phi_{d_{A_0}}\ox \Phi_{d_{B_0}})}{2d_{A_0}d_{B_0}}.
\end{equation}
It can be checked that these solutions are also feasible for dual SDP. Hence, we conclude that
\begin{equation*}
P^{\PPT}_{\suc,e}\Big(\cD^{p},\cD^{q}; k, \frac{1}{2}\Big)  = P_{\suc}\Big(\cD^p,\cD^q; \frac{1}{2}\Big) = \frac{1}{2} +  \frac{(p - q)(d_{A_0}^2d_{B_0}^2-1)}{2d_{A_0}^2d_{B_0}^2}.
\end{equation*}
\end{proof}

This result is physically intuitive because the bipartite depolarizing channel acts globally and symmetrically across Alice and Bob's joint system; the optimal discrimination strategy does not require them to establish any further non-local correlations. However, the situation changes drastically when we consider point-to-point depolarizing channels, where the noise acts purely locally on a system transmitted from Alice to Bob. First, we can also have the symmetry reduction for the point-to-point case based on Proposition~\ref{prop:bi_depo_LP} as follows.

\begin{proposition}[Symmetry reduction for point-to-point depolarizing channels]\label{prop:pp_depo_LP}
Let $\cD_{A\to B}^p$ and $\cD_{A\to B}^q$ be two depolarizing channels with distinct noise parameters $p$ and $q$, respectively where $d_{A} = d_{B} = d$. The maximal $k$-ebit-assisted success probability of discriminating $\cD_{A\to B}^p$ and $\cD_{A\to B}^q$ via a $k$-injectable PPT tester is given by the following linear program.
\begin{equation}
\begin{aligned}
P^{\PPT}_{\suc,e}(\cD_{A\to B}^p, \cD_{A\to B}^q; k, \frac{1}{2}) =\max &\; \frac{1}{2} + \frac{1}{2} (p - q) \left(d - \frac{1}{d}\right) (x - y)\\
{\rm s.t.}
&\;\; x, y, u, v \in [0, 1/d],\\
&\;\; \lambda_{\pm} \coloneqq y \pm \frac{x - y}{d}, \qquad
\mu_{\pm} \coloneqq v \pm \frac{u - v}{d},\\
&\;\; (1 - k) \mu_{\pm} \leq \lambda_{\pm} \leq (1 + k) \mu_{\pm},\\
&\;\; (1 - k)\left( \frac{1}{d} - \mu_{\pm} \right) \leq \frac{1}{d} - \lambda_{\pm} \leq (1 + k)\left( \frac{1}{d} - \mu_{\pm} \right).
\end{aligned}
\end{equation}
\end{proposition}
\begin{proof}
See Appendix~\ref{appendix:proof_pp_depo}.
\end{proof}

Utilizing Proposition~\ref{prop:pp_depo_LP}, we demonstrate that the PPT entanglement cost for a pair of distinct point-to-point depolarizing channels is nonzero. Moreover, surprisingly, one ebit suffices for optimal discrimination between any two different point-to-point depolarizing channels of arbitrary dimension via PPT testers.

\begin{proposition}[PPT entanglement cost of point-to-point depolarizing channels]\label{prop:en_cost_ppdepo}
For any $d$-dimensional point-to-point depolarizing channels $\cD^p_{A\to B},\cD_{A\to B}^q$ with distinct noise parameters $p,q$ and equiprobable priors, the one-shot PPT entanglement cost of minimum-error discrimination is $1$ ebit.
\end{proposition}
\begin{proof}
First, we will show that when $k=2$,
\begin{equation*}
    P^{\PPT}_{\suc,e}\Big(\cD^p, \cD^q; 2,\frac{1}{2}\Big) = P_{\suc}\Big(\cD^p,\cD^q; \frac{1}{2}\Big) = \frac{1}{2} +  \frac{(p - q)(d^2-1)}{2d^2},
\end{equation*}
which means $k=2$ is sufficient. Next, we will show that when $k=1$,
\begin{equation*}
    P^{\PPT}_{\suc,e}\Big(\cD^p, \cD^q; 1,\frac{1}{2}\Big) < P^{\PPT}_{\suc,e}\Big(\cD^p, \cD^q; 2,\frac{1}{2}\Big),
\end{equation*}
which indicates that $k=2$ is necessary. 

When $k=2$, for Proposition \ref{prop:pp_depo_LP}, let
\begin{equation}
    x \coloneqq \frac{1}{d}, \quad y \coloneqq 0, \quad u = v \coloneqq \frac{1}{d^2}.
\end{equation}
It follows that 
\begin{equation*}
    \lambda_{\pm} = \pm \frac{1}{d^2},\quad \mu_{\pm} = \frac{1}{d^2},
\end{equation*}
and constitute a feasible solution since
\begin{equation*}
    -\frac{1}{d^2} \le \pm \frac{1}{d^2} \le \frac{3}{d^2}, ~-\Big(\frac{1}{d}-\frac{1}{d^2}\Big) \le \frac{1}{d} \mp \frac{1}{d^2} \le 3 \Big(\frac{1}{d} - \frac{1}{d^2}\Big).
\end{equation*}
Thus we have
\begin{equation*}
    P^{\PPT}_{\suc,e}\Big(\cD^p, \cD^q; 2,\frac{1}{2}\Big) \ge \frac{1}{2} +\frac{(p-q)(d^2-1)}{2d^2}.
\end{equation*}
Further, we consider the dual SDP~\eqref{sdp:sucprob_dual_diamond}. Let 
\begin{equation}
    C_{AB} \coloneqq \frac{(p-q)(d^2-1)}{2d} \Phi_d, ~\alpha \coloneqq \frac{(p-q)(d^2-1)}{2d^2}.
\end{equation}
It can be checked that these solutions are feasible, and we have
\begin{equation*}
    P_{\suc}\Big(\cD^p,\cD^q; \frac{1}{2}\Big) \le  \frac{1}{2} +  \frac{(p - q)(d^2-1)}{2d^2}.
\end{equation*}
Since $P_{\suc}(\cD^p,\cD^q;1/2) \ge P^{\PPT}_{\suc,e}\Big(\cD^p, \cD^q; 2,\frac{1}{2}\Big)$, we have
\begin{equation*}
    P^{\PPT}_{\suc,e}\Big(\cD^p, \cD^q; 2,\frac{1}{2}\Big) = P_{\suc}\Big(\cD^p,\cD^q;\frac{1}{2}\Big) = \frac{1}{2} +  \frac{(p - q)(d^2-1)}{2d^2}.
\end{equation*}
When $k=1$, let 
\begin{equation}
    x \coloneqq \frac{1}{d}, \quad y \coloneqq \frac{1}{d(d+1)},\quad u \coloneqq \frac{1}{d},\quad v \coloneqq \frac{d+2}{2d(d+1)}.
\end{equation}
It follows that
\begin{equation*}
    \lambda_{\pm} = \frac{1 \pm 1}{d(d+1)},\quad \mu_{\pm} = \frac{d+2 \pm 1}{2d(d+1)},
\end{equation*}
and constitute a feasible solution since
\begin{equation*}
0 \le \frac{1 \pm 1}{d(d+1)} \le \frac{d+2 \pm 1}{d(d+1)},\quad 
0 \le \left( \frac{1}{d} - \frac{1 \pm 1}{d(d+1)} \right) \le 2\left( \frac{1}{d} - \frac{d+2 \pm 1}{2d(d+1)} \right).
\end{equation*}
Thus, we have 
\begin{equation*}
    P^{\PPT}_{\suc,e}\Big(\cD^p, \cD^q;1,\frac{1}{2}\Big) \geq \frac{1}{2} +  \frac{(p - q)(d-1)}{2d}.
\end{equation*}
Now consider the dual SDP~\eqref{sdp:kebit_sucprob_dual} for point-to-point channels
\begin{equation*}
\begin{aligned}
    P^{\PPT}_{\suc,e}(\cN, \cM; k, \lambda) = \min &\;\; 1-\lambda + \alpha + \beta \\
    {\rm s.t.}
    &\;\; C_{AB}, ~D_{AB}, ~E_{AB},~G_{AB}, ~H_{AB}, ~K_{AB} \ge 0, \\
    &\;\; \lambda J_{AB}^{\cN} - (1-\lambda) J_{AB}^{\cM} \le H_{AB}^{\pT_B} - K_{AB}^{\pT_B} + G_{AB}^{\pT_B} - E_{AB}^{\pT_B} + C_{AB},\\
    &\;\; (1+k)(G_{AB}^{\pT_B}-K_{AB}^{\pT_B}) + (1-k)(H_{AB}^{\pT_B}-E_{AB}^{\pT_B}) \le D_{AB}, \\
    &\; \tr_{B}(C_{AB}+H_{AB}-K_{AB}) \le \alpha \idop_A,\\
    &\; \tr_{B} \big[ D_{AB}+(k-1)H_{AB}+(k+1)K_{AB} \big] \le \beta \idop_A.
\end{aligned}
\end{equation*}
Construct
\begin{equation}
    \alpha \coloneqq \frac{(p - q)(d-1)}{2d}, ~C_{AB} \coloneqq \frac{(p - q)(d-1)}{2}\Phi_d, ~E_{AB} \coloneqq \frac{(p - q)}{2d}(\idop - F).
\end{equation}
It can be checked that these solutions are feasible, and we have
\begin{equation*}
    P^{\PPT}_{\suc,e}\Big(\cD^p, \cD^q;1,\frac{1}{2}\Big) \le  \frac{1}{2}+ \frac{(p - q)(d-1)}{2d}.
\end{equation*}
Thus, we conclude that
\begin{equation*}
    P^{\PPT}_{\suc,e}\Big(\cD^p, \cD^q;1,\frac{1}{2}\Big) = \frac{1}{2}+ \frac{(p - q)(d-1)}{2d}.
\end{equation*}
Hence, we have shown that
\begin{equation*}
    P^{\PPT}_{\suc,e}\Big(\cD^p, \cD^q; 1,\frac{1}{2}\Big) < P^{\PPT}_{\suc,e}\Big(\cD^p, \cD^q; 2,\frac{1}{2}\Big) = P_{\suc}\Big(\cD^p,\cD^q; \frac{1}{2}\Big),
\end{equation*}
which completes the proof.
\end{proof}

\begin{remark}
A channel $\cN_{A\to B}$ is said
to be PPT if $\pT\circ \cN$ is a quantum channel, where $\pT: \cL(\cH_B)\to \cL(\cH_B)$ is the transpose map~\cite{Horodecki_2000}. Equivalently, $\cN_{A\to B}$ is PPT if and only if its Choi operator is PPT, i.e., $(J_{AB}^{\cN})^{\pT_B} \geq 0$. Notably, PPT channels can distribute bound entanglement but never distillable entanglement. 
A point-to-point depolarizing channel $\cD^p_{A\to B}(X_A) = (1-p)X_A + p \tr(X_A) \frac{\idop_A}{d_A}$ is PPT if and only if $p\geq \frac{d}{1+d}$. Proposition~\ref{prop:en_cost_ppdepo} demonstrates that for any two PPT depolarizing channels $\cD^p_{A\to B}$ and $\cD^q_{A\to B}$ with $p,q \geq \frac{d}{1+d}$, optimal discrimination requires nontrivial entanglement resources, even though these channels themselves cannot generate non-positive partial transpose (NPT) entanglement.
\end{remark}

\subsubsection{Depolarized SWAP channel}\label{sec:depo_swap}
Bipartite communication channels that actively exchange information between Alice and Bob naturally exhibit cross-covariance. A paradigmatic example of such dynamics is the depolarized SWAP channel, which applies a global SWAP operation subject to uniform white noise, i.e.,
\begin{equation}
    \cS^p_{A_0B_0\to A_1B_1}(X) = (1-p) F_{A_0B_0} X F_{A_0B_0} + p \tr(X)\frac{\idop_{A_1B_1}}{d_{A_1B_1}},
\end{equation}
where $F_{A_0B_0}$ is the SWAP operator between $A_0$ and $B_0$. In the following, we apply Proposition~\ref{prop:sym_redu_covariance} to reduce the SDP formulation for a pair of distinct depolarized SWAP channels to a tractable linear program.

\begin{proposition}[Symmetry reduction for depolarized SWAP channels]\label{prop:depo_swap_LP}
Let $\cS^{p}_{A_0B_0\to A_1B_1}$ and $\cS_{A_0B_0\to A_1B_1}^{q}$ be two depolarized SWAP channels with distinct noise parameters $p$ and $q$, respectively, where $d_A=d_B = d$. The maximal $k$-ebit-assisted success probability of discriminating $\cS^{p}_{A_0B_0\to A_1B_1},\cS_{A_0B_0\to A_1B_1}^{q}$ via a $k$-injectable PPT tester is given by the following linear program.
\begin{equation*}
P_{suc, e}^{\PPT}(\cS^p, \cS^q; k, \frac{1}{2}) = \max \frac{1}{2} + \frac{p-q}{2} \left[ \left(d^2 - \frac{1}{d^2}\right)w_1 - \frac{d^2-1}{d^2}(w_2 + w_3) - \frac{(d^2-1)^2}{d^2}w_4 \right]    
\end{equation*}
subject to $w_i, q_i \in [0, 1/d^2]$ for $i \in \{1, 2, 3, 4\}$, and for all $s_1, s_2 \in \{+1, -1\}$:
\begin{itemize}
    \item[i).] $(1-k)\mu_{s_1, s_2}(Q) \le \lambda_{s_1, s_2}(W) \le (1+k)\mu_{s_1, s_2}(Q)$,
    \item[ii).] $(1-k)\left[\frac{1}{d^2} - \mu_{s_1, s_2}(Q)\right] \le \frac{1}{d^2} - \lambda_{s_1, s_2}(W) \le (1+k)\left[\frac{1}{d^2} - \mu_{s_1, s_2}(Q)\right]$,
\end{itemize}
where the eigenvalue functions $\lambda_{s_1, s_2}(W)$ and $\mu_{s_1, s_2}(Q)$ are defined by
\[
\lambda_{s_1, s_2}(W) \coloneqq w_4 + s_1\frac{w_2 - w_4}{d} + s_2\frac{w_3 - w_4}{d} + s_1 s_2\frac{w_1 - w_2 - w_3 + w_4}{d^2},
\]
and $\mu_{s_1, s_2}(Q)$ takes the exact same form substituting $w_i$ with $q_i$.
\end{proposition}
\begin{proof}
See Appendix~\ref{appendix:proof_depo_swap}.
\end{proof}

By analytically evaluating the linear program derived in Proposition~\ref{prop:depo_swap_LP}, we prove that the exact PPT entanglement cost for discriminating depolarized SWAP channels is $1$ ebit for arbitrary dimension.

\begin{proposition}[PPT entanglement cost of depolarized SWAP channels]
For any $d$-dimensional depolarized SWAP channels $\cS^p_{A_0B_0 \to A_1B_1}$ and $\cS^q_{A_0B_0 \to A_1B_1}$ with distinct noise parameters $p > q$ and equiprobable priors, the one-shot PPT entanglement cost of minimum-error discrimination is $1$ ebit.
\end{proposition}
\begin{proof}
First, we will show that when $k=2$, 
\begin{equation*}
P_{\suc, e}^{\PPT}\Big(\cS^p, \cS^q; 2, \frac{1}{2}\Big) = P_{\suc}\Big(\cS^p, \cS^q;\frac{1}{2}\Big) = \frac{1}{2} + \frac{p-q}{2}\frac{d^4-1}{d^4},
\end{equation*}
which means $k=2$ is sufficient to achieve the global optimum. Next, we will show that when $k=1$, the success probability is strictly bounded by $\frac{1}{2} + \frac{p-q}{2}\frac{d^2-1}{d^2}$, indicating that $k=2$ is necessary.

When $k=2$, in the linear program of Proposition~\ref{prop:depo_swap_LP}, let
\begin{equation*}
w_1 \coloneqq \frac{1}{d^2}, \quad w_2 = w_3 = w_4 \coloneqq 0,\quad q_1 = q_2 = q_3 = q_4 \coloneqq \frac{1}{d^4}.
\end{equation*}
Substituting these values into the eigenvalue functions yields
\begin{equation*}
    \lambda_{s_1, s_2}(W) = s_1 s_2 \frac{1}{d^4}, \quad \mu_{s_1, s_2}(Q) = \frac{1}{d^4} \quad \forall s_1, s_2 \in \{+1, -1\}.
\end{equation*}
We then verify the feasibility of this assignment. The boundary constraints $0 \le w_i \le 1/d^2$ and $0 \le q_i \le 1/d^2$ are satisfied since $1/d^4 \le 1/d^2$. For $k=2$, the primary PPT constraints $(1-k)\mu_{s_1, s_2} \le \lambda_{s_1, s_2} \le (1+k)\mu_{s_1, s_2}$ evaluate exactly to
\[
-\frac{1}{d^4} \le \pm \frac{1}{d^4} \le \frac{3}{d^4},
\]
which holds. The complementary PPT constraints $(1-k)\left(\frac{1}{d^2} - \mu_{s_1, s_2}\right) \le \frac{1}{d^2} - \lambda_{s_1, s_2} \le (1+k)\left(\frac{1}{d^2} - \mu_{s_1, s_2}\right)$ become
\[-\frac{d^2-1}{d^4} \le \frac{1}{d^2} \mp \frac{1}{d^4} \le \frac{3(d^2-1)}{d^4}
\]
For $\lambda = 1/d^4$, the inequality $-\frac{d^2-1}{d^4} \le \frac{d^2-1}{d^4} \le \frac{3(d^2-1)}{d^4}$ holds for any $d \ge 1$. For $\lambda = -1/d^4$, the inequality $-\frac{d^2-1}{d^4} \le \frac{d^2+1}{d^4} \le \frac{3(d^2-1)}{d^4}$ simplifies on the right-hand side to $d^2+1 \le 3d^2-3 \implies 2d^2 \ge 4$, which holds for any valid local dimension $d \ge 2$. Thus, this constitutes a feasible solution that reaches the maximum possible trace objective.

When $k=1$, the PPT constraints reduce to $0 \le \lambda_{s_1, s_2} \le 2\mu_{s_1, s_2}$. This strictly requires $\lambda_{s_1, s_2}(W) \ge 0$ for all sign combinations. Specifically, evaluating the cross-sign cases $\lambda_{+-}$ and $\lambda_{-+}$ gives
\begin{align*}
\lambda_{+-} &= w_4 + \frac{w_2 - w_4}{d} - \frac{w_3 - w_4}{d} - \frac{w_1 - w_2 - w_3 + w_4}{d^2}\\
&= \frac{1}{d^2}\Big[(d^2-1)w_4 + d(w_2 - w_3) + w_2 + w_3 - w_1\Big] \ge 0
\end{align*}
and
\begin{align*}
\lambda_{-+} &= w_4 - \frac{w_2 - w_4}{d} + \frac{w_3 - w_4}{d} - \frac{w_1 - w_2 - w_3 + w_4}{d^2}\\
&= \frac{1}{d^2}\Big[(d^2-1)w_4 + d(w_3 - w_2) + w_2 + w_3 - w_1\Big] \ge 0.
\end{align*}
Summing these two inequalities yields
$$\lambda_{+-} + \lambda_{-+} = \frac{2(d^2-1)w_4 + 2(w_2 + w_3) - 2w_1}{d^2} \ge 0,$$
which provides the strict algebraic requirement
$w_2 + w_3 + (d^2-1)w_4 \ge w_1.$
We now evaluate the trace objective under this condition. Factoring the objective function
$$\tr[W(J_{\bfA\bfB}^{\cS^q} - J_{\bfA\bfB}^{\cS^p})] = (p-q) \frac{d^2-1}{d^2} \left[ (d^2+1)w_1 - (w_2 + w_3) - (d^2-1)w_4 \right].$$
Applying the necessary condition $-(w_2 + w_3) - (d^2-1)w_4 \le -w_1$, we strictly bound the objective by
$$\tr[W(J_{\bfA\bfB}^{\cS^q} - J_{\bfA\bfB}^{\cS^p})] \le (p-q) \frac{d^2-1}{d^2} \left[ (d^2+1)w_1 - w_1 \right] = (p-q)(d^2-1)w_1\leq (p-q)\frac{d^2-1}{d^2}.$$
where we used the condition $w_1 \le 1/d^2$ in the inequality. Therefore, the maximal success probability for $k=1$ is strictly bounded by
$$P_{\suc, e}^{\PPT}\left(\cS^p, \cS^q; 1, \frac{1}{2}\right) \le \frac{1}{2} + \frac{p-q}{2} \frac{d^2-1}{d^2}.$$
Since $\frac{d^2-1}{d^2} = \frac{d^4-d^2}{d^4} < \frac{d^4-1}{d^4}$ for any $d \ge 2$, the unassisted $k=1$ strategy is strictly suboptimal. Thus, 1 ebit of entanglement is necessary to close the gap.
\end{proof}

\subsection{Werner-Holevo channel}\label{sec:werner}
Beyond covariant channels, we further consider another class of fundamental channels, the Werner-Holevo channels~\cite{Werner2002}. For an arbitrary dimension $d\geq 2$, the Werner-Holevo channels, denoted as $\cW_0$ and $\cW_1$, are defined via the transpose map by
\begin{equation}
\cW_0(\rho) = \frac{1}{d+1}\left[\tr(\rho) \idop + \rho^{\mathsf{T}}\right] \quad\text{and}\quad \cW_1(\rho) = \frac{1}{d-1}\left[\tr(\rho) \idop - \rho^{\mathsf{T}}\right].
\end{equation}
Leveraging the primal and dual SDP formulations established in Theorem~\ref{thm:sdp_k_inject_PPT}, we precisely quantify the entanglement resources required for this quantum channel discrimination task.
\begin{proposition}
For the $d$-dimensional Werner-Holevo channels $\cW_0$ and $\cW_1$, the one-shot PPT entanglement cost of minimum-error discrimination is $\log_2 d$ ebits.
\end{proposition}

\begin{proof}
First, we show that when $k=d$, we have that $P^{\PPT}_{\suc,e}(\cW_0, \cW_1; d, \lambda)=1$. 
The Choi operators of the Werner-Holevo channels are given by
\begin{equation}
J^{\cW_0}_{AB} = \frac{1}{d+1}(\idop_{AB} + F_{AB}),\quad J^{\cW_1}_{AB} = \frac{1}{d-1}(\idop_{AB} - F_{AB}),
\end{equation}
where $d_A=d_B=d$. Let us construct 
\begin{equation}\label{Eq:fea_WQ}
W_{AB}^{(0)} \coloneqq \frac{1}{2d}(\idop_{AB} + F_{AB}),~ W_{AB}^{(1)} \coloneqq \frac{1}{2d}(\idop_{AB} - F_{AB}),~
Q_{AB}^{(0)} = Q_{AB}^{(1)} \coloneqq \frac{\idop_{AB}}{2d},~\sigma_A = \rho_A \coloneqq \frac{\idop_A}{d}.
\end{equation}
We have that
\begin{equation*}
(1-\lambda) + \Big[\lambda\tr (W_{AB}^{(0)}J^{\cW_0}_{AB}) - (1-\lambda)\tr (W_{AB}^{(0)}J^{\cW_1}_{AB})\Big] = (1-\lambda) + \lambda\Big[\frac{1}{d(d+1)}\tr (\idop_{AB}+F_{AB})\Big] = 1,
\end{equation*}
and
\begin{equation*}
(W_{AB}^{(0)})^{\pT_B} = \frac{1}{2d}(\idop_{AB} + d\Phi_{AB}),~(W_{AB}^{(1)})^{\pT_B} = \frac{1}{2d}(\idop_{AB} - d\Phi_{AB}).
\end{equation*}
Since $\idop_{AB} \ge \Phi_{AB}$ and $1-k = 1-d < 0$, it can be checked that
\begin{align*}
    &\rho_A\ox\idop_B - (1+d)(\sigma_A\ox\idop_B - Q_{AB}^{(0)})^{\pT_B} = \frac{(1-d)\idop_{AB}}{2d} \le \frac{(\idop_{AB} + d\Phi_{AB})}{2d} = (W_{AB}^{(0)})^{\pT_B},\\
    &\rho_A\ox\idop_B - (1-d)(\sigma_A\ox\idop_B - Q_{AB}^{(0)})^{\pT_B} = \frac{(1+d)\idop_{AB}}{2d} \ge \frac{\idop_{AB} + d\Phi_{AB}}{2d} = (W_{AB}^{(0)})^{\pT_B}.
\end{align*}
Thus, the operators in Eq.~\eqref{Eq:fea_WQ} constitute a feasible solution to SDP~\eqref{sdp:kebit_sucprob} which yields that $P^{\PPT}_{\suc,e}(\cW_0, \cW_1; d, \lambda)\ge 1$. As $P^{\PPT}_{\suc,e}(\cW_0, \cW_1; d, \lambda ) \le 1$ by definition, we conclude that \[P^{\PPT}_{\suc,e}(\cW_0, \cW_1; d, \lambda)= 1.\]
In the following, we shall show that $P^{\PPT}_{\suc,e}(\cW_0, \cW_1; k, \lambda)< 1$ for any $2\leq k < d$, by considering the dual SDP~\eqref{sdp:kebit_sucprob_dual}
\begin{equation*}
\begin{aligned}
    P^{\PPT}_{\suc,e}(\cN, \cM; k, \lambda) = \min &\;\; 1-\lambda + \alpha + \beta \\
    {\rm s.t.}
    &\;\; C_{AB}, ~D_{AB}, ~E_{AB},~G_{AB}, ~H_{AB}, ~K_{AB} \ge 0, \\
    &\;\; \lambda J_{AB}^{\cN} - (1-\lambda) J_{AB}^{\cM} \le H_{AB}^{\pT_B} - K_{AB}^{\pT_B} + G_{AB}^{\pT_B} - E_{AB}^{\pT_B} + C_{AB},\\
    &\;\; (1+k)(G_{AB}^{\pT_B}-K_{AB}^{\pT_B}) + (1-k)(H_{AB}^{\pT_B}-E_{AB}^{\pT_B}) \le D_{AB}, \\
    &\; \tr_{B}(C_{AB}+H_{AB}-K_{AB}) \le \alpha \idop_A,\\
    &\; \tr_{B} \big[ D_{AB}+(k-1)H_{AB}+(k+1)K_{AB} \big] \le \beta \idop_A.
\end{aligned}
\end{equation*}
\textbf{Case 1:} $0 \le \lambda \le (d+1)/2d$. Let $C=D=E=K=0$ and
\begin{equation}\label{eq:fea_case1}
    \beta = (k-1)\alpha \coloneqq \frac{\lambda(k-1)(k+1)}{k(d+1)},~ G\coloneqq \frac{\lambda(k-1)d\Phi_{AB}}{k(d+1)} \ge 0, ~H\coloneqq \frac{\lambda(k+1)d\Phi_{AB}}{k(d+1)} \ge 0.
\end{equation}
We will verify that these operators constitute a feasible solution to the dual SDP above. 
Noticing $(1+k)G^{\pT_B} + (1-k)H^{\pT_B} = 0,~\idop_{AB} \ge F_{AB}$, and $0 \le \lambda \le (d+1)/2d$, we have that
\begin{equation*}
    G^{\pT_B} + H^{\pT_B} = \frac{2\lambda}{d+1}F_{AB} \ge \frac{(d-2\lambda+1)F_{AB}-(d+1-2d\lambda )\idop_{AB}}{(d+1)(d-1)} = \lambda J^{\cW_0}_{AB} - (1-\lambda) J^{\cW_1}_{AB}.
\end{equation*}
Further, notice that
\begin{equation*}
\begin{aligned}
    &\tr_B (C_{AB}+H_{AB}-K_{AB}) = \tr_B H_{AB} = \frac{\lambda(k+1)}{k(d+1)}\idop_A = \alpha \idop_A, \\ &\tr_{B} \big[ D_{AB}+(k-1)H_{AB}+(k+1)K_{AB} \big] = \tr_B[(k-1)H_{AB}] = \frac{\lambda(k-1)(k+1)}{k(d+1)}\idop_A = \beta \idop_A.
\end{aligned} 
\end{equation*}
Hence, the operators in Eq.~\eqref{eq:fea_case1} constitute a feasible solution to the dual, which yields that when $k < d$, 
\begin{equation*}
    P^{\PPT}_{\suc,e}(\cW_0, \cW_1; k, \lambda) \le 1-\lambda + \frac{\lambda(k+1)}{(d+1)} < 1.
\end{equation*}
\textbf{Case 2:} $(d+1)/2d <  \lambda  \le 1$, let $D=E=K=0$ and 
\begin{equation}\label{eq:fea_case2}
\begin{aligned}
\alpha &\coloneqq \frac{(2d\lambda-d-1)k+(1-\lambda)(k+1)}{k(d-1)},~\beta \coloneqq \frac{(1-\lambda)(k+1)(k-1)}{k(d-1)}, \\
C &\coloneqq \frac{(2d\lambda-1-d)(\idop_{AB}+F_{AB})}{(d-1)(d+1)}, ~G\coloneqq\frac{(1-\lambda)(k-1)d\Phi_{AB}}{k(d-1)}, ~H \coloneqq\frac{(1-\lambda)(k+1)d\Phi_{AB}}{k(d-1)}.
\end{aligned}
\end{equation}
We will also verify that these operators constitute a feasible solution to the dual SDP above. 
Noticing $(1+k)G^{\pT_B} + (1-k)H^{\pT_B} = 0,~\idop_{AB} \ge F_{AB}$, and $(d+1)/2d <  \lambda  \le 1$, we have that
\begin{equation*}
    C + G^{\pT_B} + H^{\pT_B} = \frac{(d+1-2\lambda)F_{AB} + (2d\lambda-1-d)\idop_{AB}}{(d-1)(d+1)} = \lambda J^{\cW_0}_{AB} - (1-\lambda) J^{\cW_1}_{AB}.
\end{equation*}
Further, notice that
\begin{equation*}
\begin{aligned}
    &\tr_B(C_{AB}+H_{AB}-K_{AB}) = \tr_B(C_{AB} + H_{AB}) = \frac{(2d\lambda-d-1)k+(1-\lambda)(k+1)}{k(d-1)}\idop_A = \alpha \idop_A, \\
    &\tr_{B} \big[ D_{AB}+(k-1)H_{AB}+(k+1)K_{AB} \big] = \tr_B[(k-1)H_{AB}] = \frac{(1-\lambda)(k+1)(k-1)}{k(d-1)}\idop_A = \beta \idop_A.
\end{aligned}
\end{equation*}
Hence, the operators in Eq.~\eqref{eq:fea_case2} constitute a feasible solution to the dual, which yields that when $k<d$, 
\begin{equation*}
    P^{\PPT}_{\suc,e}(\cW_0, \cW_1; k, \lambda) \le 1 - \lambda + \frac{2d\lambda-d-\lambda+k(1-\lambda)}{d-1} < 1.
\end{equation*}
Hence, we complete the proof.
\end{proof}

It is known that for two $d$-dimensional Werner--Holevo channels, the one-shot LOCC entanglement cost is also $\log_2 d$~\cite{Puzzuoli2017}. By recovering this exact cost, our analysis reveals that PPT testers do not exhibit any advantage over LOCC testers in this regime, showing that the PPT testers serve as a strictly tight approximation.

\subsection{Amplitude damping channel}
We further show some numerical results on the amplitude damping channels $\cA_{\gamma}(\cdot) = K^{}_0(\cdot)K_0^\dag + K^{}_1(\cdot)K_1^\dag$ from Alice to Bob, defined by the Kraus operators
\begin{equation}
K_0 = \left(\begin{array}{cc}
    1 & 0 \\
    0 & \sqrt{1-\gamma}
\end{array}\right),\quad
K_0 = \left(\begin{array}{cc}
    0 & \sqrt{\gamma} \\
    0 & 0
\end{array}\right).
\end{equation}
Suppose the unknown channel is drawn from the equiprobable ensemble $\{(\frac{1}{2}, \cA_{\gamma}), (\frac{1}{2}, \cA_{1-\gamma})\}$. We first compute the globally optimal average success probability (cf. Eq.~\eqref{Eq:Diamond}) for one, two, and three parallel uses of the channel. To evaluate the performance of local strategies, we then compute the maximal $k$-ebit-assisted success probability achieved by PPT testers using the SDP formulated in Theorem~\ref{thm:sdp_k_inject_PPT}, specifically benchmarking the unassisted regime ($k=1$, corresponding to 0 ebits) against the entanglement-assisted regime ($k=2$, corresponding to 1 ebit).

\begin{figure}[t]
    \centering
    \includegraphics[width=1.0\linewidth]{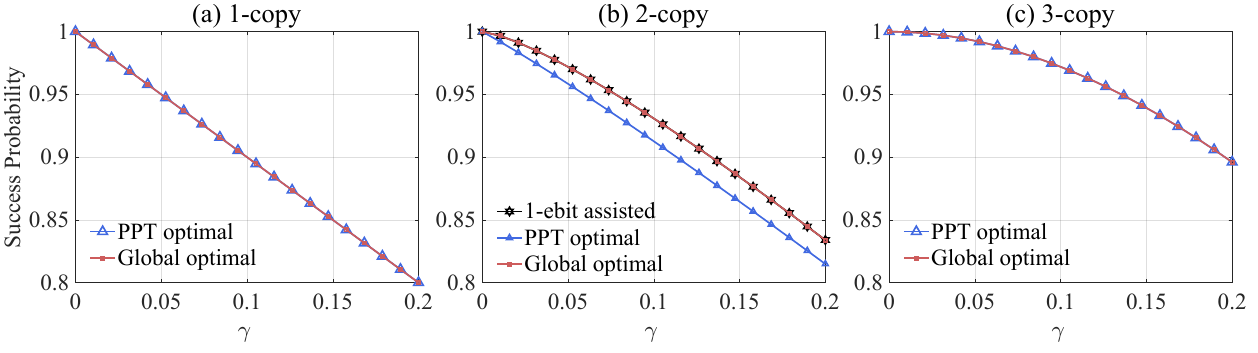}
    \caption{Maximal average success probability for discriminating the equiprobable ensemble $\{(\frac{1}{2}, \cA_{\gamma}), (\frac{1}{2}, \cA_{1-\gamma})\}$ of amplitude damping channels over the parameter range $\gamma \in [0,0.2]$. Panels (a), (b), and (c) correspond to single, two, and three parallel uses of the unknown channel. We benchmark the globally optimal success probability against the performance of PPT testers in the unassisted regime ($k=1$, $0$ ebits) and the entanglement-assisted regime ($k=2$, $1$ ebit). While single and three parallel uses achieve the global optimum without any shared entanglement, the two-copy scenario strictly requires $1$ ebit to close the performance gap between unassisted local strategies and the global optimum.}
    \label{fig:ad_numerics}
\end{figure}

From~\Cref{fig:ad_numerics}, we observe that the entanglement requirements for optimally discriminating the ensemble $\{(\frac{1}{2}, \cA_{\gamma}), (\frac{1}{2}, \cA_{1-\gamma})\}$ with $\gamma \in [0,0.2]$ depend heavily on the number of channel uses. For single and three parallel uses of the amplitude damping channel, PPT testers achieve the globally optimal success probability without any shared entanglement. In contrast, for the two-copy scenario, entanglement becomes strictly necessary, and exactly $1$ ebit is sufficient to recover the global optimum via a PPT tester.

\section*{Acknowledgments}
The authors thank Mario Berta for helpful discussions. S.\ H.\ thanks Mario Berta for hosting him in the Quantum Information Group at RWTH Aachen University during the 2025-2026 winter semester.  C.\ Z., S.\ H., and X.\ W. were partially supported by the National Key R\&D Program of China (Grant No.~2024YFB4504004), the National Natural Science Foundation of China (Grant No.~92576114, 12447107), the Guangdong Provincial Quantum Science Strategic Initiative (Grant No.~GDZX2403008, GDZX2503001), the CCF-Tencent Rhino-Bird Open Research Fund, the Guangdong Provincial Key Lab of Integrated Communication, Sensing and Computation for Ubiquitous Internet of Things (Grant No.~2023B1212010007), and X Program from HKUST(Guangzhou). G. \ K. acknowledges support from the Excellence Cluster -- Matter and Light for Quantum Computing (ML4Q-2) and funding by the European Research Council (ERC Grant Agreement No.\ 948139). 

\bibliographystyle{alpha}
\bibliography{main}

\appendix

\section{Proof of Proposition~\ref{prop:classical_torus}}\label{appendix:proof_class}
\begin{proof}
Let $\mathbb{T}^{d_{A_0}} \cong \{\mathrm{diag}(e^{i\theta_1},\cdots,e^{i\theta_{d_{A_0}}}):\theta \in [0,2\pi)\}$ and
\[
\cG \coloneqq \mathbb{T}^{d_{A_0}}\times \mathbb{T}^{d_{B_0}}\times \mathbb{T}^{d_{A_1}}\times \mathbb{T}^{d_{B_1}},
\]
equipped with the normalized Haar measure $\mu$ (product of the uniform measures on each torus factor).
For $g=(\theta^{A_0},\theta^{B_0},\theta^{A_1},\theta^{B_1})\in \cG$, define diagonal unitaries
\begin{align*}
U_{A_0}(g) &\coloneqq \sum_{a_0} e^{i\theta^{A_0}_{a_0}}\ketbra{a_0}{a_0},&
V_{B_0}(g) &\coloneqq \sum_{b_0} e^{i\theta^{B_0}_{b_0}}\ketbra{b_0}{b_0},\\
U_{A_1}(g) &\coloneqq \sum_{a_1} e^{i\theta^{A_1}_{a_1}}\ketbra{a_1}{a_1},&
V_{B_1}(g) &\coloneqq \sum_{b_1} e^{i\theta^{B_1}_{b_1}}\ketbra{b_1}{b_1}.
\end{align*}
Define the unitary representation on $A_0B_0A_1B_1$ used in Proposition~\ref{prop:sym_redu_covariance} by
\begin{equation}\label{eq:Gamma_def_nodimmatch}
\Gamma_{\bfA\bfB}(g)
\;\coloneqq\;
\overline U_{A_0}(g)\ox \overline V_{B_0}(g)\ox U_{A_1}(g)\ox V_{B_1}(g),
\end{equation}
where the overline denotes entrywise complex conjugation in the chosen bases. Since $\Gamma_{\bfA\bfB}(g)$ is diagonal in the basis $\{\ket{a_0b_0a_1b_1}\}$, it commutes with every diagonal operator in that basis.
For classical channel $\cN$ and $\cM$, $J^{\cN}$ and $J^{\cM}$ are diagonal in $\{\ket{a_0b_0a_1b_1}\}$, hence for all $g\in G$ and $\cC\in\{\cN,\cM\}$,
\begin{equation}\label{eq:J_invariant_nodimmatch}
\Gamma_{\bfA\bfB}(g)\,J^{\cC}\,\Gamma_{\bfA\bfB}(g)^\dagger = J^{\cC}.
\end{equation}
Therefore $\cN$ and $\cM$ satisfy the covariance assumption of Proposition~\ref{prop:sym_redu_covariance} with respect to $\cG$ and the diagonal unitary families $\{U_{A_0}(g)\}$ and $\{V_{B_0}(g)\}$. 
Consider the \emph{unrestricted} (global) discrimination SDP for $P_{\suc}(\cN,\cM;\lambda)$. Applying the same twirling construction to an optimal unrestricted solution yields an optimal unrestricted solution $(W^\star,\rho^\star)$ that is diagonal.
By diagonality, $(W^\star)^{\pT_{\bfB}}=W^\star\ge 0$ and $(\rho^\star\ox \idop - W^\star)^{\pT_{\bfB}}=\rho^\star\ox \idop - W^\star\ge 0$, so $(W^\star,\rho^\star)$ is feasible to the PPT-restricted program and achieves the global optimum. Consequently,
$P^{\PPT}_{\suc,e}(\cN,\cM;1,\lambda)=P_{\suc}(\cN,\cM;\lambda)$, and thus the one-shot PPT entanglement cost is $0$ ebits.
\end{proof}

\section{Proof of Proposition~\ref{prop:bi_depo_LP}}\label{appendix:proof_bi_depo}
\begin{proof}
For depolarizing channel $\cD^p_{A_0B_0\to A_1B_1}$, we can choose \[\Gamma_{\bfA\bfB} \coloneqq \overline{U}_{A_0} \ox \overline{V}_{B_0} \ox U_{A_1}\ox V_{B_1}\] for $U \in \cU(\cH_{A_0})$ and $V\in\cU(\cH_{B_0})$ in Proposition~\ref{prop:sym_redu_covariance}. Consequently, we can assume the optimal solution  $(W_{\bfA\bfB},Q_{\bfA\bfB},\rho_{A_0B_0},\sigma_{A_0B_0})$ of SDP~\eqref{sdp:kebit_sucprob} for $\cD_{A_0B_0\to A_1B_1}^p,\cD^q_{A_0B_0\to A_1B_1}$ satisfy
\begin{equation}
    \rho_{A_0B_0} = \sigma_{A_0B_0} = \frac{\idop}{d_{A_0B_0}},~~W_{\bfA\bfB}, Q_{\bfA\bfB} \in \operatorname{span}\left\{ P_1,P_2,P_3,P_4 \right\},
\end{equation} 
where the orthogonal projectors spanning the local commutant are given by
\begin{align*}
    &P_1 = \Phi_{A_0A_1} \otimes \Phi_{B_0B_1},  & P_2 = \Phi_{A_0A_1} \otimes (\idop_{B_0B_1} - \Phi_{B_0B_1}),\\
    &P_3 = (\idop_{A_0A_1} - \Phi_{A_0A_1}) \otimes \Phi_{B_0B_1}, & P_4 = (\idop_{A_0A_1} - \Phi_{A_0A_1}) \otimes (\idop_{B_0B_1} - \Phi_{B_0B_1}).
\end{align*}
We can parameterize $W_{\bfA\bfB}= \sum_{i=1}^4 w_i P_i,~Q_{\bfA\bfB} =  \sum_{i=1}^4 q_i P_i$. The operator constraint $0 \le W_{\bfA\bfB} \le \frac{\idop_{\bfA\bfB}}{d_{A_0B_0}}$ yields the scalar bounds
\begin{equation}\label{eq:boundary_bidepo}
    0 \le w_i \le \frac{1}{d_{A_0B_0}},\quad 0 \le q_i \le \frac{1}{d_{A_0B_0}},\quad \forall i\in \{0,1,2,3\}.
\end{equation}
Noticing that
\begin{equation}
W_{\bfA\bfB}^{\pT_{\bfB}} = W_{\bfA\bfB}, \quad
Q_{\bfA\bfB}^{\pT_{\bfB}} = Q_{\bfA\bfB},
\end{equation}
the operator inequality $(1-k)Q_{\bfA\bfB}^{\pT_{\bfB}} \le W_{\bfA\bfB}^{\pT_{\bfB}} \le (1+k)Q_{\bfA\bfB}^{\pT_{\bfB}}$ strictly decouples into four independent sets of scalar inequalities
\begin{equation}\label{eq:PPT_bidepo1}
(1-k)q_i \le w_i \le (1+k)q_i \quad \forall i \in \{1,2,3,4\}.
\end{equation}
Similarly, since $\rho_{A_0B_0}^{\pT_{B_0}} \otimes \idop_{A_1B_1} = \frac{\idop_{\bfA\bfB}}{d_{A_0B_0}}$, the complementary PPT constraint identically decouples into
\begin{equation}\label{eq:PPT_bidepo2}
(1 - k)\left( \frac{1}{d_{A_0B_0}} - q_i \right) \leq \frac{1}{d_{A_0B_0}} - w_i \leq (1 + k)\left( \frac{1}{d_{A_0B_0}} - q_i \right) \quad \forall i \in \{1,2,3,4\}.
\end{equation}
Now consider the objective function. Note that
\[
J_{\bfA\bfB}^{\cD^q} - J^{\cD^p}_{\bfA\bfB} = (p-q) \left( d_{A_0B_0} P_1 - \frac{\idop_{\bfA\bfB}}{d_{A_0B_0}} \right).
\]
Substituting $\tr P_1 = 1$, $\tr P_2 = d_{B_0}^2-1$, $\tr P_3 = d_{A_0}^2-1$, and $\tr P_4 = (d_{A_0}^2-1)(d_{B_0}^2-1)$) yields the objective function:
\begin{equation*}
\tr\!\Big[ W_{\bfA\bfB} \left( J_{\bfA\bfB}^{\cD^q} - J_{\bfA\bfB}^{\cD^p} \right) \Big] \!= (p - q) \left[ w_1 \frac{d_{A_0B_0}^2-1}{d_{A_0B_0}} - w_2 \frac{d_{B_0}^2-1}{d_{A_0B_0}} - w_3 \frac{d_{A_0}^2-1}{d_{A_0B_0}} - w_4 \frac{(d_{A_0}^2-1)(d_{B_0}^2-1)}{d_{A_0B_0}} \right].
\end{equation*}
Assuming $p > q$, maximizing this objective requires maximizing $w_1$ while minimizing $w_2$, $w_3$, and $w_4$. Crucially, because the PPT constraints in Eq.~\eqref{eq:PPT_bidepo1},\eqref{eq:PPT_bidepo2} and boundary constraints Eq.~\eqref{eq:boundary_bidepo} are completely decoupled and identical for each $i$, the independent optimizations will naturally select a maximal value for $w_1 \coloneqq x$ and a minimal, identical value for $w_2 = w_3 = w_4 \coloneqq y$. It follows that $W_{\bfA\bfB}, Q_{\bfA\bfB}$ can be parameterized by two pairs of coefficients $(x,y)$ and $(u,v)$, respectively, i.e., \[W_{\bfA\bfB} = x\Phi_{d_{A_0}}\ox \Phi_{d_{B_0}} + y(\idop_{\bfA\bfB} - \Phi_{d_{A_0}}\ox \Phi_{d_{B_0}}),~~Q_{\bfA\bfB} = u\Phi_{d_{A_0}}\ox \Phi_{d_{B_0}} + v(\idop_{\bfA\bfB} - \Phi_{d_{A_0}}\ox \Phi_{d_{B_0}}).\] 
Substituting these variables directly recovers the linear program stated in the proposition.
\end{proof}

\section{Proof of Proposition~\ref{prop:pp_depo_LP}}\label{appendix:proof_pp_depo}
\begin{proof}
Since the depolarizing channel $\cD_{A\to B}(X_A) = (1-p)X_A + p \tr(X_A) \frac{\idop_A}{d_A} $ is covariant under the full unitary group $\cU(\cH_A)$ where $\cH_A\cong \cH_B$ and $d_A=d_B=d$, we can choose $\Gamma_U \coloneqq \overline{U} \ox U$ for $U \in \cU(\cH_A)$ in Proposition~\ref{prop:sym_redu_covariance}. Consequently, we can assume the optimal solution  $(W_{AB},Q_{AB},\rho_{A},\sigma_{A})$ of SDP~\eqref{sdp:kebit_sucprob} for $\cD^p,\cD^q$ satisfy
\begin{equation}
    \rho_A = \sigma_A = \frac{\idop}{d},~~W_{AB}, Q_{AB} \in \operatorname{span}\left\{ \Phi_d,\ \idop - \Phi_d \right\}.
\end{equation}
It follows that $W_{AB}, Q_{AB}$ can be parameterized by two pairs of coefficients $(x,y)$ and $(u,v)$, respectively, i.e., $W_{AB} = x\Phi_d + y(\idop_{d^2} - \Phi_d)$ and $Q_{AB} = u\Phi_d + v(\idop_{d^2} - \Phi_d)$. It follows that $0 \leq x, y, u, v \leq 1/d$ since $0 \leq W_{AB}, Q_{AB} \leq \idop_{AB}/d$. We can calculate that
\begin{equation*}
\tr\left[ W_{AB} \left( J_{AB}^{\cD^q} - J_{AB}^{\cD^p} \right) \right]
= (p - q) \, \tr\left[ W_{AB} \left( d\Phi_d - \frac{\idop_{AB}}{d} \right) \right]
= (p - q) \left(d - \frac{1}{d}\right) (x - y).
\end{equation*}
Recall $\Phi_d^{\pT_B} = \frac{1}{d} F$. Hence we have
\begin{equation*}
W_{AB}^{\pT_B} = y \, \idop + \frac{x - y}{d} F, \qquad
Q_{AB}^{\pT_B} = v \, \idop + \frac{u - v}{d} F.
\end{equation*}
Since $F$ has eigenvalues $\pm 1$ on the symmetric and antisymmetric subspaces, the inequalities $(1 - k) Q_{AB}^{\pT_B} \leq W_{AB}^{\pT_B} \leq (1 + k) Q_{AB}^{\pT_B}$ are equivalent to the two scalar pairs
\begin{equation}
(1 - k) \mu_{\pm} \leq \lambda_{\pm} \leq (1 + k) \mu_{\pm},
\end{equation}
with $\lambda_{\pm}, \mu_{\pm}$ as claimed. The complementary bounds involving $\rho_A \ox \idop - W_{AB}^{\pT_B}$ give the second displayed pair. All other constraints have already been taken into account. This completes the reduction.
\end{proof}

\section{Proof of Proposition~\ref{prop:depo_swap_LP}}\label{appendix:proof_depo_swap}
\begin{proof}
The Choi operator of a depolarized SWAP channel $\cS^{p}_{A_0B_0\to A_1B_1}$ is given by
\begin{equation*}
    J_{\bfA\bfB}^{\cS^p} = (1-p)d^2 \Phi_{A_0B_1}\ox \Phi_{B_0A_1} + p\frac{\idop_{\bfA\bfB}}{d_{\bfA\bfB}}.
\end{equation*}
Consider the group $\cG = \cU(\cH_d) \times \cU(\cH_d)$ and the representation $\Gamma_{\bfA\bfB} = \overline{U}_{A_0}\ox \overline{V}_{B_0} \ox V_{A_1} \ox U_{B_1}$ for all $U,V\in \cU(\cH_d)$. We have 
\[
\Gamma_{\bfA\bfB} J^\cC_{\bfA\bfB} \Gamma_{\bfA\bfB} = J^\cC_{\bfA\bfB},~\forall \cC\in\{\cS^{p}_{A_0B_0\to A_1B_1},\cS^{q}_{A_0B_0\to A_1B_1}\}.
\]
By Schur's Lemma, the commutant of $\Gamma_{AB}(G)$ decomposes into four orthogonal projectors spanning the cross-paired subsystems $(A_0, B_1)$ and $(B_0, A_1)$:
\begin{align*}
    & P_1 = \Phi_{A_0B_1} \ox \Phi_{B_0A_1}, &P_2 = \Phi_{A_0B_1} \ox (\idop_{B_0A_1} - \Phi_{B_0A_1})\\
    & P_3 = (\idop_{A_0B_1} - \Phi_{A_0B_1}) \ox \Phi_{B_0A_1}, &P_4 = (\idop_{A_0B_1} - \Phi_{A_0B_1}) \ox (\idop_{B_0A_1} - \Phi_{B_0A_1}).
\end{align*}
Any invariant operator $W_{\bfA\bfB}$ can be parameterized as $W_{\bfA\bfB} = \sum_{i=0}^3 w_i P_i$. The constraint $0 \le W_{\bfA\bfB} \le \rho_{A_0B_0} \ox \idop_{A_1B_1} = \frac{\idop_{\bfA\bfB}}{d^2}$ yields the bounds $0 \le w_i \le \frac{1}{d^2}$ for all $i$. We parameterize $Q_{\bfA\bfB} = \sum_{i=0}^3 q_i P_i$ identically. Note that 
\[J_{\bfA\bfB}^{\cS^q} - J_{\bfA\bfB}^{\cS^p} = (p-q) \left( d^2 P_1 - \frac{\idop_{\bfA\bfB}}{d^2} \right).\] 
Since $\idop_{\bfA\bfB} = P_1 + P_2 + P_3 + P_4$ and $\tr P_1 = 1, \tr P_2 = \tr P_3 = d^2-1$, and $\tr P_4 = (d^2-1)^2$, we obtain
\begin{equation*}
    \tr [W_{\bfA\bfB}(J_{\bfA\bfB}^{\cS^q} - J_{\bfA\bfB}^{\cS^p})] = (p-q) \left[ w_1 \left(d^2 - \frac{1}{d^2}\right) - (w_2 + w_3)\frac{d^2-1}{d^2} - w_4\frac{(d^2-1)^2}{d^2} \right].
\end{equation*}
Applying $\pT_B$ to the four projectors yields
\begin{align*}
& P_1^{\pT_{\bfB}} = \frac{1}{d^2} F_{A_0B_1} \ox F_{B_0A_1}, & P_2^{\pT_{\bfB}} = \frac{1}{d} F_{A_0B_1} \ox \left(\idop_{B_0A_1} - \frac{1}{d} F_{B_0A_1}\right), \\
& P_3^{\pT_{\bfB}} = \left(\idop_{A_0B_1} - \frac{1}{d} F_{A_0B_1}\right) \ox \frac{1}{d} F_{B_0A_1}, & P_4^{\pT_{\bfB}} = \left(\idop_{A_0B_1} - \frac{1}{d} F_{A_0B_1}\right) \ox \left(\idop_{B_0A_1} - \frac{1}{d} F_{B_0A_1}\right).    
\end{align*}
Because the SWAP operators $F$ have eigenvalues $s = +1$ on the symmetric subspace and $s = -1$ on the antisymmetric subspace, the operator $W_{\bfA\bfB}^{\pT_{\bfB}}$ jointly diagonalizes into four distinct eigenspaces uniquely characterized by the signs $s_1, s_2 \in \{+1, -1\}$. Notice that 
\begin{itemize}
    \item the coefficient for $\idop_{A_0B_1} \otimes \idop_{B_0A_1}$ is $w_4$;
    \item the coefficient for $\idop_{A_0B_1} \otimes F_{B_0A_1}$ is $\frac{w_3 - w_4}{d}$;
    \item the coefficient for $F_{A_0B_1} \otimes \idop_{B_0A_1}$ is $\frac{w_2 - w_4}{d}$;
    \item the coefficient for $F_{A_0B_1} \otimes F_{B_0A_1}$ is $\frac{w_1 - w_2 - w_3 + w_4}{d^2}$,
\end{itemize}
and $F_{A_0B_1},F_{B_0A_1}$ can be simultaneously diagonalized since they commute. Therefore, the operator inequality $(1-k)Q_{\bfA\bfB}^{\pT_{\bfB}} \le W_{\bfA\bfB}^{\pT_{\bfB}} \le (1+k)Q_{\bfA\bfB}^{\pT_{\bfB}}$  is strictly equivalent to enforcing the corresponding scalar inequalities identically on each of the four eigenspaces. The complementary bounds involving $\rho_{A_0B_0}^{\pT_{B_0}} \otimes \idop_{A_1B_1} - W_{\bfA\bfB}^{\pT_{\bfB}}$ identically translate to the second pair of scalar inequalities, recognizing that $\rho_{A_0B_0}^{\pT_{B_0}} \ox \idop_{A_1B_1} = \frac{\idop_{\bfA\bfB}}{d^2}$. This completes the reduction.
\end{proof}
\end{document}